%% file: main.tex
\documentclass[11pt,letterpaper,oneside]{article}
\usepackage[margin=1in]{geometry}
\usepackage[sort&compress,numbers]{natbib}
\usepackage{graphicx}
\usepackage{amsthm,amsmath,amssymb}
\usepackage[usenames, dvipsnames]{xcolor}
\usepackage{booktabs}

\graphicspath{{figures}}

\usepackage{nicefrac}
\usepackage{xspace}
\usepackage{paralist}
\usepackage{enumitem}
\usepackage{mathtools}
\usepackage{relsize}
\usepackage[algo2e,noend,ruled,linesnumbered]{algorithm2e}
\usepackage{hyperref} 
\hypersetup{
	colorlinks=true,  
	linkcolor=RoyalBlue,   
	filecolor=magenta, 
	urlcolor=RoyalBlue,    
	citecolor=OliveGreen  
}
\usepackage[capitalize]{cleveref}
\usepackage{dsfont}
\usepackage{thmtools} 
\usepackage{thm-restate}
\usepackage{multirow}
\usepackage{subcaption}
\usepackage{authblk}

\DeclareMathOperator*{\argmin}{argmin}

\newtheorem{theorem}{Theorem}
\newtheorem{problem}{Problem}
\newtheorem{definition}{Definition}
\newtheorem{corollary}{Corollary}
\newtheorem{lemma}{Lemma}

\definecolor{MyGreen}{HTML}{66c2a5}
\definecolor{MyRed}{HTML}{fc8d62}
\definecolor{MyBlue}{HTML}{8da0cb}
\SetNlSty{bfseries}{\color{black}}{}

 \input{text/notation}
\begin{document}

\title{\LARGE{Model-Agnostic Approximation of Constrained~Forest~Problems}}

\author[1,2,3]{Corinna Coupette}
\author[4]{Alipasha Montaseri}
\author[5]{Christoph Lenzen}
\affil[1]{Aalto University}
\affil[2]{KTH Royal Institute of Technology}
\affil[3]{Max Planck Institute for Informatics}
\affil[4]{Sharif University of Technology}
\affil[5]{CISPA Helmholtz Center for Information Security}

\date{}
\maketitle
\thispagestyle{empty}

\begin{abstract}
	\input{text/abstract}
\end{abstract}

\clearpage
\pagestyle{plain}
\setcounter{page}{1}
\input{text/introduction}

\input{text/preliminaries}

\input{text/overview}

\input{text/conclusion}

\clearpage

\section*{Acknowledgments}

Part of this work was done while C.C. was a research associate and A.M. was an intern at the Max Planck Institute for Informatics.  
C.C. was also supported by \emph{Digital Futures} at KTH Royal Institute of Technology.

\vfill

\appendix
 {\noindent\bfseries\LARGE Appendix}
 \phantomsection{}\label{sec:apx-start}
 \addcontentsline{toc}{section}{\protect\numberline{Z}Appendix}
 \vspace*{1em}
 
 \noindent In this appendix, we provide the following supplementary materials.
 \begin{itemize}[label=, leftmargin=\widthof{\textbf{H.\hspace*{0.5em}}},align=right,noitemsep]
 	\item[\textbf{\ref{apx:notation}.}]
 	Tabular overview of our notation. 
 	\item[\textbf{\ref{apx:centralized-algorithm}.}] 
 	Derivation and verification of our model-agnostic shell-decomposition algorithm. 
 	\item[\textbf{\ref{apx:distributed}.}]
 	 Results for proper forest functions in the distributed setting. 
 	\item[\textbf{\ref{apx:parallel}.}]
 	 Results for proper forest functions in the parallel setting. 
 	\item[\textbf{\ref{apx:streaming}.}]
 	 Results for proper forest functions in the streaming setting. 
 	 \item[\textbf{\ref{apx:lower}.}]
 	 Deterministic lower bounds for Steiner Forest input specifications in \Congest.
 	 \item[\textbf{\ref{apx:msf}.}]
 	 Deterministic Minimum-Spanning-Forest construction with Partwise Aggregation in \Congest.
 	 \item[\textbf{\ref{apx:related}.}]
 	 Extended discussion of related work.
 \end{itemize}

\input{text/appendix-notation}

\clearpage

\input{text/appendix-algorithms}

\clearpage

\input{text/appendix-distributed}
\input{text/appendix-parallel}
\input{text/appendix-streaming}
\input{text/appendix-lowerbounds}

\input{text/appendix-msf}
\input{text/appendix-related}

\phantomsection
\addcontentsline{toc}{section}{\protect\numberline{Z}References}
\bibliographystyle{ACM-Reference-Format}
\bibliography{bibliography}

\end{document}

%% file: text/notation.tex
\renewcommand{\epsilon}{\varepsilon}

\newcommand{\streaming}{Multi-Pass Streaming\xspace}
\newcommand{\streamingabbrv}{MPS\xspace}

\newcommand{\graphtask}{Disjoint Aggregation\xspace}
\newcommand{\graphparameter}{\ensuremath{W}\xspace}

\newcommand{\pacomplexity}{\ensuremath{T^{PA}}\xspace}
\newcommand{\dacomplexity}[1]{\ensuremath{T^{DA(#1)}}\xspace}
\newcommand{\disjointagg}{\ensuremath{\mathrm{DA}}\xspace}
\newcommand{\stepthree}{Root-Path Selection\xspace}
\newcommand{\facloclong}{Facility Placement and Connection\xspace}
\newcommand{\facloc}{FPC\xspace}
\newcommand{\stepthreeabbrv}{RPS\xspace}
\newcommand{\naturals}{\ensuremath{\mathds{N}_0}\xspace}
\newcommand{\naturalsnozero}{\ensuremath{\mathds{N}}\xspace}

\newcommand{\reals}{\ensuremath{\mathds{R}}\xspace}
\newcommand{\positiveintegers}[1]{\ensuremath{[#1]}\xspace}
\newcommand{\nonnegativeintegers}[1]{\ensuremath{[#1]_0}\xspace}
\newcommand{\requests}{\ensuremath{\mathcal{R}}\xspace}
\newcommand{\costsum}{\ensuremath{s}\xspace}

\newcommand{\GW}{GW\xspace}
\newcommand{\myfunction}{\ensuremath{f}\xspace}
\newcommand{\growth}{\ensuremath{\eta}\xspace}
\newcommand{\cardinality}[1]{|#1|}

\newcommand{\overgrowth}{\ensuremath{\epsilon''}\xspace}
\newcommand{\rootnode}{\ensuremath{r}\xspace}
\newcommand{\distanceeps}{\ensuremath{\epsilon'}\xspace}

\newcommand{\graph}{\ensuremath{G}\xspace}

\newcommand{\nodes}{\ensuremath{V}\xspace}
\newcommand{\edges}{\ensuremath{E}\xspace}

\newcommand{\nedges}{\ensuremath{m}\xspace}
\newcommand{\nnodes}{\ensuremath{n}\xspace}
\newcommand{\somepath}{\ensuremath{p}\xspace}
\newcommand{\shortestpath}{\ensuremath{P}\xspace} 
\newcommand{\shortestpathdiameter}{\ensuremath{s}\xspace}

\newcommand{\terminal}{\ensuremath{\tau}\xspace}

\newcommand{\terminals}{\ensuremath{\mathcal{T}}\xspace}
\newcommand{\distance}{\ensuremath{d}\xspace}
\newcommand{\hopdistance}{\ensuremath{h}\xspace}
\newcommand{\ncomponents}{\ensuremath{k}\xspace}

\newcommand{\nterminals}{\ensuremath{t}\xspace}
\newcommand{\forest}{\ensuremath{F}\xspace}

\newcommand{\lowerbound}{\ensuremath{LB}\xspace}
\newcommand{\edge}{\ensuremath{e}\xspace}
\newcommand{\node}{\ensuremath{v}\xspace}
\newcommand{\othernode}{\ensuremath{u}\xspace}
\newcommand{\cost}{\ensuremath{c}\xspace}
\newcommand{\components}{\ensuremath{\mathcal{C}}\xspace}
\newcommand{\component}{\ensuremath{C}\xspace}
\newcommand{\radius}{\ensuremath{r}\xspace}
\newcommand{\yvar}{\ensuremath{y}\xspace}
\newcommand{\xvar}{\ensuremath{x}\xspace}
\newcommand{\thesubset}{\ensuremath{S}\xspace}
\newcommand{\edgesubset}{\ensuremath{Z}\xspace}
\newcommand{\cut}{\ensuremath{\delta}\xspace}

\newcommand{\tree}{\ensuremath{T}\xspace}
\newcommand{\augmentation}{\ensuremath{A}\xspace}

\newcommand{\inputcomponent}{\ensuremath{\Lambda}\xspace}
\newcommand{\componentidentifier}{\ensuremath{\lambda}\xspace}

\newcommand{\sources}{\ensuremath{X}\xspace}
\newcommand{\targets}{\ensuremath{Y}\xspace}

\newcommand{\BO}{\mathcal{O}\xspace}
\newcommand{\Congest}{\textsc{Congest}\xspace}

\newcommand{\hopdiameter}{\ensuremath{D}\xspace}
\newcommand{\shortcutquality}{\ensuremath{Q}\xspace}
\newcommand{\myfunctioncomplexity}{\ensuremath{FFE}\xspace}
\newcommand{\tildeO}{\ensuremath{\widetilde{\mathcal{O}}}\xspace}
\newcommand{\tildeo}{\ensuremath{\widetilde{o}}\xspace}
\newcommand{\tildeOmega}{\ensuremath{\widetilde{\Omega}}\xspace}

\newcommand{\reducedcost}{\ensuremath{c'}\xspace}
\newcommand{\costreduction}{\ensuremath{\bar{c}}\xspace}
\newcommand{\reducedcomponents}{\ensuremath{\components_\ballunion}\xspace}

\newcommand{\mergeedges}{\ensuremath{M}\xspace}
\newcommand{\ballunion}{\ensuremath{U}\xspace}
\newcommand{\nactive}{\ensuremath{a}\xspace}

%% file: text/abstract.tex

\noindent\emph{Constrained Forest Problems} (CFPs) 
as introduced by Goemans and Williamson in $1995$
capture a wide range of network design problems with edge subsets as solutions, 
such as Minimum Spanning Tree, Steiner Forest, and Point-to-Point Connection. 
While individual CFPs have been studied extensively in individual computational models, 
a unified approach to solving general CFPs in multiple computational models has been lacking. 
Against this background, we present the \emph{shell-decomposition algorithm}, 
a \emph{model-agnostic meta-algorithm} that efficiently computes a $(2+\epsilon)$-approximation to CFPs for a broad class of forest functions.

To demonstrate the power and flexibility of this result, 
we instantiate the shell-decomposition algorithm 
for three fundamental, NP-hard CFPs 
(Steiner Forest, Point-to-Point Connection, and Facility Placement and Connection)
in three different computational models (\Congest, PRAM, and \streaming). 
For example, for constant $\epsilon$, we obtain the following $(2+\epsilon)$-approximations in the \Congest model:
\begin{enumerate}[nosep,label=(\arabic*)]
	\item For Steiner Forest specified via input components (SF--IC), 
	where each node knows the identifier of one of $k$ disjoint subsets of $V$ (the input components), 
	we achieve a deterministic $(2+\epsilon)$-approximation in $\tildeO(\sqrt{\nnodes}+\hopdiameter+\ncomponents)$ rounds, where $\hopdiameter$ is the hop diameter of the~graph.
	\item   
	For Steiner Forest specified via symmetric connection requests (SF--SCR), 
	where connection requests are issued to pairs of nodes $\othernode,\node\in \nodes$, 
	we leverage randomized equality testing to reduce the running time to $\tildeO(\sqrt{\nnodes}+\hopdiameter)$,
	succeeding with high probability.
	\item For Point-to-Point Connection, 
	we provide a $(2+\epsilon)$-approximation in $\tildeO(\sqrt{\nnodes}+\hopdiameter)$ rounds.
	\item For \facloclong, a relative of non-metric Uncapacitated Facility Location, 
	we obtain a $(2+\varepsilon)$-approximation in $\tildeO(\sqrt{\nnodes} + \hopdiameter)$ rounds.
\end{enumerate}
We further show how to replace the $\sqrt{\nnodes}+\hopdiameter$ term by the complexity of solving Partwise Aggregation, achieving (near-)universal optimality in any setting in which a solution to Partwise Aggregation in near-shortcut-quality time is known.

%% file: text/introduction.tex
\section{Introduction}

The classic approach to determining the computational complexity of a task 
consists in deriving upper and lower bounds that are as tight as possible in a particular model of computation.
On the upper-bound side, this can result in solutions that are specifically engineered toward the chosen computational model, 
such that the task at hand needs to be reexamined for every relevant model. 
Worse still, one might end up with fragile solutions that overcome only the challenges specifically represented by the model under study.
At first glance, lower bounds might appear more robust in this regard, since they highlight an obstacle that \emph{any} algorithm needs to overcome.
However, many such lower bounds are shown for very specific network topologies that have never been observed in practice.
This might lead to a false sense of success in having classified the complexity of a given task---%
although the lower bound merely indicates that the parameters used to capture the computational complexity of the task are insufficient. 
In brief, traditional upper and lower bounds bear two limitations: \emph{model specificity} and \emph{existential optimality}.

A textbook illustration of these limitations is provided by the Steiner Forest problem (SF), 
which generalizes the well-known Steiner Tree problem (ST).
In the classic SF formulation using input components (SF--IC), 
we are given a weighted graph, along with disjoint subsets of nodes $\nodes_1,\ldots,\nodes_\ncomponents\subseteq \nodes$,
and the goal is to determine a minimum-weight subgraph spanning each $\nodes_i$, ${i\in [\ncomponents]}$.
This general and fundamental connectivity problem has been studied in depth in the classic centralized model of computation~\cite{agrawal1995trees,goemans1995approximation,jain2001approximation,chekuri2008approximate,gupta2015greedy}, and the state of the art in the \Congest model is due to \citet{lenzen2014steiner}.
Their algorithms are modified and adapted variants of a well-known centralized algorithm by \citet{agrawal1995trees}, 
but they are specifically tailored to the \Congest model. 
One might suspect that their underlying algorithmic techniques could be transferred to other computational models, 
but this cannot be readily determined from their specialized solution (\emph{model~specificity}).
Furthermore, even their fastest algorithm runs for $\tildeOmega(\sqrt{\nnodes})$ rounds---%
which is known to be necessary in the worst case (\emph{existential optimality}): 
The Minimum Spanning Tree (MST) problem is a special case of the Steiner Forest Problem with $\ncomponents=1$ and $\nodes_1=\nodes$, 
to which the $\tildeOmega(\sqrt{\nnodes})$ lower bound by \citet{dassarma2012distributed} applies.
However, no real-world networks remotely similar to the lower-bound graph are known, 
and the technique of using \emph{low-congestion shortcuts} has been demonstrated to overcome the $\sqrt{\nnodes}$ barrier for MST in many settings~\cite{ghaffari2016algorithms,haeupler2021shortcuts,ghaffari2021congestion}.
These findings motivate us to revisit the Steiner Forest problem, 
along with a much broader class of connectivity problems, 
in an attempt to overcome the abovementioned limitations. 

\subsection*{Beyond Model Specificity and Existential Optimality}
To tackle the challenges of \emph{model specificity} and \emph{existential optimality}, two paradigms have emerged.

First, 
numerous works have pushed toward what we term \emph{model agnosticism}, 
developing algorithms 
that can be readily instantiated in a wide range of computational models \cite[e.g.,][]{rozhon2022paths,mukhopadhyay2020mincut,chang2019coloring,blikstad2022communication,haeupler2023flows}.
The individual steps of these algorithms 
are either core tasks like distance computation, identifying connected components, and sorting, 
which are well-studied across a wide range of models,
or they are easily implemented at low cost in any notable model.
Hence, we call these algorithms \emph{model-agnostic (meta-)algorithms}.\footnote{%
	While the term \emph{meta-algorithm} is widely used to describe algorithms that use other algorithms as building blocks (and we often drop the \emph{meta-} for brevity), 
	there is no established term to characterize algorithms working in many \emph{computational models}. 
	We propose \emph{model-agnostic} for this purpose in analogy to its usage in machine learning, 
	where it has come to describe algorithms that are compatible with a wide variety of machine-learning models~\cite{finn2017modelagnostic,yoon2018bayesian}.
} 
Naturally, model-agnostic algorithms are more robust against model variations, 
and by allowing us to plug in optimized model-specific subroutines for the core computational problems,
they sacrifice little performance over model-specific solutions.
Thus, model-agnostic algorithms can be viewed as model-agnostic reductions to more basic computational tasks. 
Moreover, since these basic tasks can be solved by model-specific subroutines,
they improve whenever progress is made on the current performance bottleneck in a given computational model.

Second, 
\emph{universal optimality}, 
which was coined in the context of distributed computing, 
pushes to design topology-adaptive algorithms that are asymptotically worst-case optimal on \emph{every} network topology, i.e., when varying input parameters \emph{other} than the underlying communication graph 
\cite[e.g.,][]{haeupler2021universally,ghaffari2022universally,zuzic2022paths}.\footnote{%
	Note that \emph{instance optimality}, 
	which requires that an algorithm is $\BO(1)$-competitive with any other always-correct algorithm on each instance, 
	including the best algorithm for the specific instance,
	is often unachievable \cite{haeupler2021universally}.
} 
One might argue that this idea is no different than taking into account more parameters, 
such as the network diameter, the node connectivity, or any other quantity meaningful for the computational task.
However, parametrizing complexity by the input graph is extremely general, 
subsuming a large number of parameters that one might consider, 
and capturing any \emph{topology-specific obstacle} for the task at hand in a given computational model. 

\subsection*{Our Contributions in Brief}

In this work, we study a general class of connectivity problems called \emph{Constrained Forest Problems} (CFPs), 
introduced by \citet{goemans1995approximation}, 
through the lenses of \emph{model agnosticism} and \emph{universal optimality}. 
Intuitively, a CFP is specified by a binary function \myfunction that indicates, 
for each node subset $\thesubset\subseteq \nodes$, 
whether it needs to be connected to the outside world, 
i.e., if the output must contain an edge from \thesubset to $\nodes\setminus \thesubset$.
For example, the Steiner Forest problem can be specified by 
$\myfunction(\thesubset)=1$ if and only if there exists some $i\in [\ncomponents]$ such that both $\nodes_i\cap \thesubset$ and $\nodes_i\cap (\nodes\setminus \thesubset)$ are non-empty, for given disjoint node subsets 
$V_i\subseteq V$.

\paragraph{Model-Agnostic Algorithm for CFPs.} 
We devise the \emph{shell-decomposition algorithm}, 
a generic approximation algorithm for CFPs that is efficient if \myfunction can be evaluated efficiently. 
\begin{restatable}[Model-Agnostic Complexity of Constrained Forest Problems]{theorem}{metatheorem}\label{thm:metatheorem}
	Given $0<\epsilon\le 1$ and a graph $\graph=(\nodes,\edges)$ with polynomially bounded edge weights $\cost\colon\edges\rightarrow\naturalsnozero$, 
	a $(2+\epsilon)$-approximation to a CFP with proper forest function $\myfunction\colon 2^\nodes \rightarrow \{0,1\}$ can (up to bookkeeping operations) be obtained at complexity $\tildeO\left((\text{aSSSP} + \text{MST} + \text{\stepthreeabbrv} + \text{FFE})\epsilon^{-1}\right)$, 
	with the terms in the sum denoting the complexities of solving 
	\begin{inparaenum}[(1)]
		\item aSSSP: $(1+\epsilon)$-approximate Set-Source Shortest-Path Forest, 
		\item MST: Minimum Spanning Tree, 
		\item \stepthreeabbrv: \stepthree, and 
		\item FFE: Forest-Function Evaluation for \myfunction.
	\end{inparaenum}
\end{restatable}

MST, aSSSP, and \stepthreeabbrv (solvable via Transshipment) are 
well-studied in many models of computation.
In contrast, while \myfunction must satisfy certain constraints (see \cref{sec:preliminaries}),~it~can still be (ab)used to force evaluating an arbitrarily hard function $g$ on the entire input.\footnote{%
	For input graph $\graph=(\nodes,\edges)$ with uniform edge weights and any function $g$, choosing an arbitrary node $\node\in \nodes$ and setting $f(\thesubset)=1$ if and only if \begin{inparaenum}[(i)]
	\item $\thesubset=\{\node\}$ or $S=\nodes\setminus\{\node\}$, and 
	\item $g(\graph)=1$
\end{inparaenum} results in a proper forest function \myfunction. 
}
Thus, \cref{thm:metatheorem} can be viewed as confining the hardness of the task arising from the choice of~\myfunction~to $\BO(\epsilon^{-1}\log \nnodes)$ iterations of evaluating
 $\myfunction(\component)$ for all $\component \in \components$, 
 where \components is a set of disjoint connected components. 

To illustrate the power of our result, 
we apply our machinery in three models of computation---%
\Congest, Parallel Random-Access Machine (PRAM), and \streaming (\streamingabbrv)---%
to three NP-hard CFPs: 
\begin{inparaenum}[(1)]
	\item \emph{Steiner Forest (SF)};
	\item \emph{Point-to-Point Connection (PPC)}, 
	i.e., given $\sources, \targets\subset \nodes$ of equal cardinality, 
	finding a lightest set of edges that balances the number of nodes from $\sources$ and $\targets$ in each induced connected component; and  
	\item \emph{\facloclong}~\emph{(\facloc)}, i.e., minimizing the cost of opening facilities at some nodes and connecting a set of clients $C\subseteq \nodes$ to them.
\end{inparaenum}
\Cref{tab:results} summarizes our results in \Congest;
for all problems, our PRAM algorithms require $\tildeO(\epsilon^{-3}\nedges)$ work and $\tildeO(\epsilon^{-3})$ depth,
and our \streamingabbrv algorithms need $\tildeO(\epsilon^{-3})$ passes and $\tildeO(\nnodes)$ memory.

\begin{table}[t]
	\centering
	\caption{%
		Overview of \Congest results for $\epsilon\in \Theta(1)$ derived from our algorithm. 
		Deterministic and randomized approximations are marked with D resp. R.
		Here, 
		\shortestpathdiameter is the shortest-path diameter, 
		\nterminals is the number~of terminals, 
		\ncomponents is the number of SF input components, 
		\shortcutquality is the shortcut quality of the input~graph, 
		and~$\pacomplexity\in \BO(\sqrt{\nnodes}+\hopdiameter)$, $\pacomplexity\ge \shortcutquality$, is the running time of Partwise Aggregation.  
		The $\nnodes^{o(1)}$ factors can be removed conditionally on the existence of certain cycle-cover algorithms~\cite{rozhon2022paths}.
	}\label{tab:results}
	\input{text/table-results}
\end{table}

\begin{table}[t]
	\centering
	\caption{%
		Overview of our \Congest results for the four input variants of SF, 
		with $\epsilon\in \Theta(1)$ and the notation of \cref{tab:results}. 
		Again, the $\nnodes^{o(1)}$ stems from the need for better cycle-cover~algorithms.
	}\label{tab:results-sf}
	\input{text/table-results-sf}
\end{table}

\paragraph{Approaching Universal Optimality.}
In our upper bounds for \Congest, $\sqrt{\nnodes}+\hopdiameter$ can largely be replaced by $\pacomplexity \nnodes^{o(1)}$, where $\pacomplexity$ is the running time of an algorithm performing Partwise Aggregation~\cite{haeupler2018faster}.
Due to the respective hardness results~\cite{haeupler2021universally}, 
this implies that the running times of our solutions to PPC and \facloc are universally optimal up to a factor of $\nnodes^{o(1)}$.

For SF--IC, this is true up to the additive term of \ncomponents, the number of input components that need to be connected.
Here, Partwise Aggregation is insufficient: 
Evaluating \myfunction requires us to determine, for each set $\nodes_i$, 
if it is contained in a single connected component induced by the set of edges that have been selected into the current (partial) solution, 
but the $\nodes_i$ may not induce connected components in \graph.
Existential lower bounds demonstrate that this obstacle is unavoidable in general~\cite{lenzen2014steiner}.
We propose a new graph parameter, the $p$-\emph{weave quality} 
$\graphparameter(p)$, and the corresponding task of \emph{Disjoint Aggregation on $p$ parts} ($\disjointagg(p)$), i.e., to perform an aggregation on each set in a given partition of \nodes into $p$ parts.
The additive $k$ in the running time of our SF algorithm for \Congest then is replaced by \dacomplexity{k}, the running time of disjoint aggregation on $k$~parts.

Last but not least, an orthogonal consideration yields surprising results, summarized in~\Cref{tab:results-sf}.
For the SF problem, the input representation drives the problem complexity.
It is known that if to encode connectivity requirements, 
nodes are given the identifiers of other nodes they must connect to in the output (SF--CR), 
an existential lower bound of $\tildeOmega(\nterminals)$ applies~\cite{lenzen2014steiner}, 
where \nterminals is the number of \emph{terminals}, 
i.e., nodes that need to be connected to some other node.
In our corresponding upper bound, the additive \ncomponents then becomes an additive \nterminals, again resulting in existential but not universal optimality.
Interestingly, this picture is turned upside down when inputs are \emph{symmetric} in the following sense: 
If $\othernode\in \nodes$ knows that it must connect to $\node\in \nodes$ by the input, then also \node knows that it must connect to \othernode (SF--SCR).
In this setting, we can exploit symmetry to efficiently evaluate \myfunction using randomized equality testing.
This leads to an algorithm running in $\epsilon^{-3}\pacomplexity n^{o(1)}$ rounds, which is close to universal optimality, provided that an efficient partwise-aggregation routine is available.
Similarly, we observe that the $\tildeOmega(\ncomponents)$ bound can be circumvented if nodes receive not only the identifier of their input component $\nodes_i$, 
but also the size $\cardinality{\nodes_i}$ of this input component (SF--CIC).
We then obtain a randomized algorithm running in $\tildeO(\epsilon^{-3}(\sqrt{\nnodes}+\hopdiameter)+\epsilon^{-1}n^{\nicefrac{2}{3}})$ time.

We note that a simple adaptation of the lower-bound construction from~\citet{lenzen2014steiner} shows the same hardness (i.e., $\tildeOmega(\nterminals)$ for SF--SCR and $\tildeOmega(\ncomponents)$ for SF--CIC, respectively) for \emph{deterministic} algorithms, 
based on a communication-complexity reduction from $2$-player equality testing.
To the best of our knowledge, this is the first natural example of a provably large gap between the randomized and deterministic complexity in the \Congest model for a \emph{global} problem.

\paragraph*{Structure.}
Having introduced our main definitions in \cref{sec:preliminaries},
we provide a technical overview of our results in \cref{sec:overview}. 
To conclude the main text, we discuss open questions in \cref{sec:conclusion}. 
All proofs, model-specific topics, and further related work are deferred to our comprehensive \hyperref[sec:apx-start]{Appendix}.

%% file: text/table-results.tex
\footnotesize
\setlength{\tabcolsep}{5pt}
\begin{tabular}{r|r|rlc|ll}
	\toprule
	Problem&\multicolumn{1}{l|}{LB}&\multicolumn{3}{l}{Previous Work}&\multicolumn{2}{|l}{Our Work}\\
	&\multicolumn{1}{r|}{Ref.}&\multicolumn{1}{l}{APX}&Complexity&Ref.&\multicolumn{1}{l}{APX}&Complexity\\
	\midrule
	\multirow{2}{*}{SF(--IC)}
		&\multirow{2}{*}{$\tildeOmega(\shortcutquality)$~\cite{haeupler2021universally}}&$(2+\epsilon)$~D & $\tildeO(\shortestpathdiameter\ncomponents+\sqrt{\min\{\shortestpathdiameter\nterminals,\nnodes\}})$
		&\multirow{2}{*}{\cite{lenzen2014steiner}}
		&\multirow{2}{*}{$(2+\epsilon)$ D}&  \multirow{2}{*}{$\tildeO(\min\{\pacomplexity\nnodes^{o(1)},\sqrt{\nnodes}+\hopdiameter\}+\ncomponents)$}\\
	&&$\BO(\log\nnodes)$~R & $\tildeO(\min\{\shortestpathdiameter,\sqrt{\nnodes}\}+\hopdiameter+\ncomponents)$&&&\\
	PPC&$\tildeOmega(\shortcutquality)$~\cite{haeupler2021universally}&\multicolumn{3}{c|}{---}&$(2+\epsilon)$~D&$\tildeO(\min\{\pacomplexity\nnodes^{o(1)},\sqrt{\nnodes}+\hopdiameter\})$\\
	\facloc&
	$\tildeOmega(\shortcutquality)$~\cite{haeupler2021universally}&\multicolumn{3}{c|}{---}
	&$(2+\epsilon)$ D&$\tildeO(\min\{\pacomplexity\nnodes^{o(1)},\sqrt{\nnodes}+\hopdiameter\})$\\
	\bottomrule
\end{tabular}

%% file: text/table-results-sf.tex
\footnotesize
\setlength{\tabcolsep}{5pt}
\begin{tabular}{r|p{0.3575\linewidth}rrr}
	\toprule
	Problem&Input&LB&APX&Complexity\\\midrule
	SF--IC&Component identifiers ${\componentidentifier\colon \nodes\rightarrow [\ncomponents]\cup\{\bot\}}$; node \node knows $\componentidentifier_\node$&$\tildeOmega(\shortcutquality+\ncomponents)$~R&$(2+\epsilon)$ D&$\tildeO(\min\{\pacomplexity\nnodes^{o(1)},\sqrt{\nnodes} + \hopdiameter\} + \ncomponents)$\\[1pt]
	SF--CIC&As in SF--IC, but node \node knows $\componentidentifier_\node$ \emph{and} $\cardinality{\{\othernode\in\nodes\mid \componentidentifier_\othernode = \componentidentifier_\node\}}$&$\tildeOmega(\shortcutquality+\ncomponents)$~D&$(2+\epsilon)$ R&$\tildeO(\nnodes^{\nicefrac{2}{3}} + \hopdiameter)$\\[1pt]
	SF--CR&Each node \node is given $\requests_\node\subseteq\nodes\setminus\{\node\}$&$\tildeOmega(\shortcutquality+\nterminals)$~R&$(2+\epsilon)$ D&$\tildeO(\min\{\pacomplexity\nnodes^{o(1)},\sqrt{\nnodes} + \hopdiameter\} + \nterminals)$\\[1pt]
	SF--SCR&$\requests\subseteq\binom{\nodes}{2}$; node \node knows\newline $\requests_\node = \{\othernode\in\nodes\mid \{\othernode,\node\}\in\requests\}$&$\tildeOmega(\shortcutquality+\nterminals)$~D&$(2+\epsilon)$ R&$\tildeO(\min\{\pacomplexity\nnodes^{o(1)},\sqrt{\nnodes} + \hopdiameter\})$\\
	\bottomrule
\end{tabular}

%% file: text/preliminaries.tex
\section{Preliminaries}
\label{sec:preliminaries}

To begin, we introduce our basic notation and 
the class of problems we are interested in. 

\paragraph{Basic Notation (\cref{tab:notation}).}
For a set \thesubset, we denote its cardinality by $\cardinality{\thesubset}$, 
its power set by ${2^\thesubset = \{X\mid X\subseteq S\}}$, 
and the set of its $k$-element subsets by $\binom{\thesubset}{k}$. 
We extend functions $\myfunction\colon \thesubset\rightarrow \reals$ to subsets $X\subset \thesubset$ in the natural way by setting $\myfunction(X) = \sum_{x\in X}\myfunction(x)$, 
and write the sets of positive and nonnegative integers no greater than $k$ as $\positiveintegers{k} = \{i\in\naturalsnozero\mid i \leq k\}$ resp.  $\nonnegativeintegers{k} = \{i\in\naturals\mid i \leq k\}$.

We consider weighted graphs $\graph = (\nodes, \edges)$ 
with $\nnodes = \cardinality{\nodes}$ nodes, 
$\nedges = \cardinality{\edges}$ edges, 
and edge weights (\emph{edge costs}) $\cost\colon\edges\rightarrow\naturals$ polynomially bounded in \nnodes.\footnote{%
	Assuming polynomially bounded edge weights allows us to encode polynomial sums of edge weights with $\BO(\log \nnodes)$ bits, 
	which means that we can encode edge weights in a single message (\Congest) or memory word (PRAM and \streamingabbrv). 
	Zero-weight edges arise naturally when simulating contractions in the distributed setting. 
	We can handle them by scaling all non-zero edge weights by $\nicefrac{\nnodes}{\epsilon}$ (where w.l.o.g., $\nicefrac{1}{\epsilon}$ is polynomially bounded as well), 
	and assigning weight $1$ to all previously zero-weight edges.\label{fn:zero-weight-handling}
} 
Each node is equipped with a unique identifier of $\BO(\log \nnodes)$ bits, 
which is also used break ties.
An $\ell$-hop \emph{path} from $\othernode\in\nodes$ to $\node\in\nodes$, 
denoted $\somepath(\othernode,\node)$, 
is a sequence of distinct edges $(\edge_1,\dots\edge_\ell)$ 
arising from a sequence of distinct nodes $(\node_1,\dots,\node_{\ell+1})$ such that $\edge_i = \{\node_i,\node_{i+1}\} \in \edges$ for $i\in[\ell]$, 
$\node_1 = \othernode$, and $\node_{\ell+1} = \node$. 
The (unweighted) \emph{hop distance} between \othernode and \node is the smallest number of hops needed to go~from \othernode to \node, 
i.e., $\hopdistance(\othernode,\node) = \min\{i\mid\exists~\somepath(\othernode,\node)~\text{with}~\cardinality{\somepath(\othernode,\node)} = i\}$, 
and the \emph{hop diameter} of \graph is $\hopdiameter = \max\{\hopdistance(\othernode,\node)\mid \{\othernode,\node\}\in\binom{\nodes}{2}\}$. 
The (weighted) \emph{shortest-path distance} between \othernode and \node is ${\distance(\othernode,\node) = \min\{i \mid \exists~\somepath(\othernode,\node)~\text{with}~\cost(\somepath(\othernode,\node)) = i\}}$.
A~\emph{shortest path} between \othernode and \node, denoted $\shortestpath(\othernode,\node)$, 
is a path from \othernode to \node of length $\distance(\othernode,\node)$.
The \emph{shortest-path diameter} \shortestpathdiameter of \graph is
the maximum over all $\{\othernode,\node\}\in\binom{\nodes}{2}$ of the minimum number of hops contained in a shortest path from \othernode to \node, 
i.e.,  ${\shortestpathdiameter = \max\{\min \{\cardinality{\shortestpath(\othernode,\node)} \mid \{\othernode,\node\}\in\binom{\nodes}{2}\}\}}$. 
Given a \emph{cut} $(\thesubset, \nodes\setminus\thesubset)$, 
the set of edges with exactly one endpoint in~$\thesubset$ is denoted as $\cut(\thesubset) = \{\edge\in\edges\mid \cardinality{\edge\cap\thesubset} = 1\}$.

In complexity statements, 
we use $\tildeO$ and $\tildeOmega$ to suppress factors of $\log^{\BO(1)}\nnodes$.
An event occurring \emph{with high probability} (w.h.p.) has probability at least $1-\nicefrac{1}{\nnodes^c}$ for a freely chosen constant $c\geq 1$.

\paragraph{Constrained Forest Problems.}
\label{subsec:problems}
We are interested in Constrained Forest Problems (CFPs) as introduced by \citet{goemans1995approximation}.\footnote{%
	To simplify the technical exposition, like \citet{goemans1995approximation}, 
	we disallow zero-weight edges. 
	However, it is straightforward to extend our approach to zero-weight edges by scaling edge weights as discussed in \cref{fn:zero-weight-handling}.
	}
Given a graph $\graph = (\nodes,\edges)$ with edge costs $\cost: \edges\rightarrow \naturalsnozero$
and a function $\myfunction: 2^\nodes \rightarrow \{0,1\}$, 
a CFP asks us to solve the integer program stated as \cref{prob:constrained-forests}, 
whose dual relaxation is provided as \cref{prob:dual-relaxation}.

\begin{minipage}[t]{0.5\linewidth}
	\begin{problem}[CFP Primal IP]
		\label{prob:constrained-forests}
		\begin{align*}
			&\min \sum_{\edge\in\edges}\cost(\edge)\xvar_\edge\\
			\text{s.t.}\quad&x(\cut(\thesubset))\geq\myfunction(\thesubset)\quad\forall~\emptyset\neq\thesubset\subset\nodes\\
			&\xvar_\edge\in\{0,1\}\quad\forall~\edge\in\edges
		\end{align*}
	\end{problem}
\end{minipage}~%
\begin{minipage}[t]{0.5\linewidth}
	\begin{problem}[CFP Dual LP]
		\label{prob:dual-relaxation}
		\begin{align*}
			&\max\sum_{\thesubset\subset\nodes} \myfunction(\thesubset)\yvar_\thesubset\\
			\text{s.t.}\quad&\sum_{\thesubset: e\in\cut(\thesubset)} \yvar_\thesubset \leq \cost(\edge)\quad\forall~\edge\in\edges\\
			&\yvar_\thesubset\geq 0\;
		\end{align*}
	\end{problem}
\end{minipage}
\vspace*{1pt}

That is, a CFP is a minimization problem 
whose optimal solution is a \emph{forest} of edges from the input graph \graph 
that meets the \emph{constraints} imposed by the \emph{forest function} \myfunction.
Like \citet{goemans1995approximation}, 
we consider CFPs with \emph{proper} functions \myfunction, which satisfy
\begin{inparaenum}[(1)]
	\item \emph{zero}, i.e., $\myfunction(\nodes) = 0$, 
	\item \emph{symmetry}, i.e.,  
	$\myfunction(\thesubset) = \myfunction(\nodes\setminus\thesubset)$, and 
	\item \emph{disjointness} (also called \emph{maximality} \cite{goemans1996primal}), i.e., 
	if $A\cap B = \emptyset$ for two sets $A$ and $B$, then $\myfunction(A) = \myfunction(B) = 0$ implies $\myfunction(A\cup B) = 0$.\footnote{%
		We also assume that \myfunction is \emph{nontrivial},
		i.e., that there exists at least one $\thesubset\subset\nodes$ such that  $\myfunction(\thesubset) = 1$.
	}
\end{inparaenum}

For a CFP with forest function \myfunction, 
a node \node is called a \emph{terminal} if $\myfunction(\{\node\}) = 1$, 
and the number of terminals $\terminals = \{\{\node\}\mid \node\in\nodes, \myfunction(\{\node\}) = 1\}$ for a given function \myfunction is denoted as ${\nterminals  = \cardinality{\terminals}}$.

\paragraph{Specific CFPs.}
To demonstrate the flexibility of our approach, we consider Survivable Network Design Problems (SNDPs) that originate from three different real-world challenges.
\begin{definition}[Steiner Forest (SF)]
Given a partition of the terminals $\terminals=\nodes_1\dot{\cup}\ldots\dot{\cup}\nodes_{\ncomponents}$, find a minimum-cost edge subset connecting the nodes in each input component.
The corresponding forest function evaluates to $1$ on $\thesubset\subseteq \nodes$ if and only if there is $i\in \positiveintegers{\ncomponents}$ such that $\emptyset \neq \nodes_i\cap \thesubset\neq \nodes_i$.
We distinguish between several input representations:
\begin{compactitem}[left=3em]
\item [SF--CR] Terminal $v$ is given the identifiers of other terminals as \emph{connection requests} $R_v\subseteq \terminals$.
\item [SF--SCR] As SF-CR, but connection requests are symmetric, i.e., $\othernode\in R_v\Leftrightarrow \node \in R_u$.
\item [SF--IC] Terminal $\node\in \nodes_i$ is given a unique identifier (of size $\BO(\log \nnodes)$) for its \emph{input component}.
\item [SF--CIC] As SF-IC, but $\node\in \nodes_i$ is also given the \emph{cardinality} $\cardinality{\nodes_i}$ of its input component as input. 
\end{compactitem}
\end{definition}
SF has practical relevance especially in infrastructure development \cite{agrawal1995trees}, 
with MST ($\terminals=\nodes_1=\nodes$) and ST ($\terminals=\nodes_1\subseteq \nodes$) as special cases.
It is NP-complete \cite{karp1972reducibility} and APX-hard~\cite{chlebik2008steiner}.

\begin{definition}[Point-to-Point Connection (PPC)]
Given a set of sources $\sources\subset \nodes$ and a set of targets $\targets\subset \nodes$, find a minimum-cost edge subset such that in each connected component, the number of sources equals the number of targets. 
That is, for $\thesubset\subseteq \nodes$, $\myfunction(\thesubset)=1$ if and only if $\cardinality{\thesubset\cap \sources}\neq\cardinality{\thesubset\cap \targets}$.
\end{definition}
PPC is motivated by challenges from circuit switching and VLSI design, and NP-complete~\cite{li1992delivery}.

Our last problem is \facloclong (\facloc), 
an NP-complete facility-location-type problem arising, e.g., in operations research. 
Intuitively, \facloc can be stated as follows.
\begin{definition}[\facloc\ {[}intuitive{]}]
	Given for each node $\node\in\nodes$ an opening cost $o_\node\in \naturalsnozero$ and indication whether it is in the set of clients $C\subseteq \nodes$,
identify a subset $O\subseteq \nodes$ of facilities to open and an edge set $F\subseteq \edges$ such that each client is connected to a facility by $F$, minimizing $\sum_{\node\in O}o_\node + \sum_{\edge\in F}\cost(\edge)$.
\end{definition}

To turn this task into a CFP matching our framework, 
we add one additional node $s\notin \nodes$ and, for each $\node\in \nodes$, an edge $\{\node,s\}$ of weight $\cost(\{\node,s\})=o_\node$.
The task then becomes to determine a (low-weight) Steiner Tree spanning $C\cup \{s\}$, i.e., the special case of SF with $k=1$.
\begin{definition}[\facloc\ {[}rephrased{]}]
	Given, for each node $\node\in\nodes$, an opening cost $o_\node\in \naturalsnozero$ and an indication whether it is in the set of clients $C\subseteq \nodes$, 
	solve ST on $\graph = (V\dot{\cup}\{s\},E\dot{\cup}\{\{v,s\}\,|\,v\in V\})$,
	with edge costs of $c(e)$ for $e\in E$ and $c(\{v,s\})=o_v$ for $v\in V$
	as well as terminals $\terminals=C\cup\{s\}$. 
\end{definition}
Even in \Congest, we can solve ordinary ST instances efficiently, 
regardless of the specific input representation.
However, the \emph{virtual node} $s$ and its incident edges need to be \emph{simulated} in the chosen model of computation.
This is trivial in PRAM (simply modify the input representation in parallel) and \streaming (since we use $\Omega(n\log n)$ bits of memory anyway).
For \Congest, we show that  we can simulate the virtual node efficiently using Partwise Aggregation.

\paragraph{Partwise Aggregation and Shortcut Quality.}
Finally, we introduce a subroutine that we use as a black box to achieve near-universal optimality in cases where efficient solutions are known.
\begin{definition}[Partwise Aggregation~\cite{ghaffari2016algorithms}]
For disjoint node sets $V_1,\ldots,V_k\subseteq V$, suppose that $V_i$ induces a connected subgraph.
In the \emph{Partwise Aggregation (PA)} problem, each $v\in V_i$ is given a unique identifier for $V_i$ (of size $\BO(\log n)$) and a second $\BO(\log n)$-bit value $x(v)\in X$. For a specified associative and commutative operator $\bigoplus\colon X\times X\to X$, for each $i\in [k]$ and each $v\in V_i$, $v$ needs to compute its output $\bigoplus_{w\in V_i}x(w)$.
\end{definition}
\begin{definition}[Shortcut Quality]
The \emph{shortcut quality} of \graph, denoted by \shortcutquality, is the maximum over all feasible operators $\bigoplus$ and partitions of $V$ of the minimum number of rounds in which a \Congest algorithm with knowledge of the full topology can solve Partwise Aggregation. Put differently, the algorithm may preprocess the graph and the operator, but must then compute the output within \shortcutquality rounds after the nodes have been given their inputs to the PA instance.
\end{definition}
\citet{haeupler2021universally} show that $\tildeOmega(\shortcutquality)$ is a lower bound for MST 
(i.e., SF and ST with $\ncomponents=1$ and $\nodes_1=\nodes$) and shortest $s$-$t$ path (the special case of PPC with $|\sources|=|\targets|=1$)---%
regardless of the approximation ratio and also for randomized Las Vegas algorithms, i.e., those that guarantee a feasible output.
They also prove that PA can be solved in $\tildeO(\shortcutquality)$ rounds in the \emph{Supported \Congest} model, 
which is \Congest with the unweighted graph topology given as part of the input.

%% file: text/overview.tex
\section{Technical Overview}\label{sec:overview}

In this section, we provide a high-level outline of our techniques and the technical challenges to be overcome.
We also use this opportunity to discuss the most relevant related work in context.

\paragraph{Model-Agnostic Algorithm for Proper Constrained Forest Problems.}

Our shell-decom\-position algorithm, 
depicted in \Cref{fig:overview},
is based on the primal-dual formulation for general CFPs given by \citet{goemans1995approximation} (GW~algorithm), 
but it can also be seen as a model-agnostic generalization of the algorithm by \citet{agrawal1995trees}, 
ported to the distributed setting by \citet{lenzen2014steiner}. 
These algorithms have also been called \emph{moat-growing} algorithms~\cite{gupta2009constant,archer2011improved,lenzen2014steiner}.

Intuitively, our algorithm operates as follows.
We maintain connected \emph{components} \components, 
initialized to the singletons $\components = \{\{\node\} \mid \node\in \nodes\}$.
A component $\component\in\components$ is \emph{active} 
if $\myfunction(\component)=1$.
The algorithm concurrently grows balls around all \emph{active components}, with respect to the metric induced by the then-current edge costs \reducedcost.
In the growth process, two balls can touch only when at least one of their associated components is active,
and only when the balls of \emph{two active components} \component, $\component'$ touch, 
adding edges to merge the components can make the resulting component inactive---%
otherwise, i.e., assuming ${\myfunction(\component\cup \component')=f(\component')=0}$, 
symmetry implies $\myfunction(\nodes\setminus (\component\cup \component'))=0$, 
disjointness $\myfunction(\nodes\setminus \component)=\myfunction((\nodes\setminus (\component\cup \component'))\cup \component')=0$, 
and using symmetry again we get $\myfunction(\component)=0$, a contradiction.

Until the balls around two components touch, 
they are disjoint, 
witnessing that the dual problem 
has a solution with weight larger than the product of the current radius times the number of currently active components.
Accordingly, when merging active components, 
we can afford to connect the terminals by adding a shortest path between them to the primal solution, 
paying a cost of at most twice the radius at merge.
Because each merge we~perform reduces the number of active components by at least one, 
the ball growth always witnesses sufficient additional weight in a dual solution to pay for future merges up to an approximation factor of~$2$.
Upon termination, i.e., when all components are inactive, 
the set of edges we selected constitutes a feasible solution.

The intuition sketched above already suggests that the ball-growing process allows for substantial concurrency.
To achieve high efficiency across the board in various models,
our shell-decomposition algorithm performs several modifications to the original GW algorithm.
\begin{enumerate}[nosep,label=(\Alph*)]
	\item \emph{Incremental Solution-Set Construction.}\label{relaxation:incremental-solution} 
	We merge active components that touch \emph{regardless} of whether this ultimately turns out to be necessary to satisfy connectivity requirements. 
	This does not impact the approximation guarantee, 
	which was implicit already in the contribution by \citet{agrawal1995trees} 
	and is now made explicit in our reformulation of the framework of \citet{goemans1995approximation}. 
	As a result,
	we need to determine if \myfunction imposes further connectivity requirements 
	only as often as we iterate through the loop in \Cref{fig:overview}.
  \item \emph{Approximate Distance Computations.}\label{relaxation:approximate-distances}  
  At the cost of a factor of $1+\epsilon$ in the approximation ratio,
  we replace exact distance computations with $(1+\epsilon)$-approximate distance computations. 
  The challenge here is to ensure that we do not violate the dual constraints when charging dual variables (corresponding to cuts) based on the progress of the primal solution. 
  This can be achieved by constructing the dual solution in the true metric space, 
  rather than reusing the approximate distances leveraged by the primal solution.
  In contrast to the other modifications, this requires a comparatively involved argument, and it is the main technical novelty and contribution in this part of our work.
  \item \emph{Deferred Forest-Function Evaluation.}\label{relaxation:deferred-updates} 
  Again at the cost of a factor of $1+\epsilon$ in the approximation ratio,
  we let components grow to radii that are integer powers of $1+\epsilon$ 
  before reevaluating our forest function to update their activity status. 
  This technique was introduced by \citet{lenzen2014steiner} for the Steiner Forest problem specifically; 
  we show its general correctness in the context of the \GW algorithm.
  Using this technique, we can limit the number of loop iterations in~\Cref{fig:overview} to $\BO(\epsilon^{-1}\log\nnodes)$ (assuming polynomially bounded edge weights).
\end{enumerate}
We prove that our model-agnostic shell-decomposition algorithm combines these changes while maintaining an approximation ratio of $2+\varepsilon$, 
which is detailed in \Cref{apx:centralized-algorithm}.

Leveraging the modifications specified in  \Cref{relaxation:approximate-distances,relaxation:deferred-updates,relaxation:incremental-solution}, 
to derive concrete algorithms in specific models of computation, 
what remains is to implement the individual steps in~\Cref{fig:overview}.
Up to simple book-keeping operations, this entails four main tasks (drawn as colored boxes in \Cref{fig:overview}):
\begin{enumerate}[nosep]
  \item \emph{(Approximate) 
  	Set-Source Shortest-Path Forest (aSSSP).} 
  This is essentially computing a single-source shortest-path tree with a virtual source node, 
  so we can plug in state-of-the-art algorithms for each model of interest~\cite{becker21stream,rozhon2022paths}. 
  \item \emph{Minimum Spanning Tree (MST).} 
  This is another well-studied task for which near-optimal solutions are known in all prominent models~\cite{pettie2002optimal,haeupler2021universally,pettie2002randomized,feigenbaum2005graph}.
  \item \emph{\stepthree (\stepthreeabbrv).} 
  Here, we are given a forest rooted at sources 
  (where each node knows its parent in its tree and its closest source), 
  along with a number of marked nodes.
  Our goal is to select the edges on the path from each marked node to its closest source.
  This problem can be straightforwardly addressed by solving a much more general flow problem: 
  \emph{(Approximate) Transshipment},
  in which flow demands are to be satisfied at minimum cost. 
  Again, we can plug in state-of-the-art algorithms for each model of interest~\cite{becker21stream,rozhon2022paths}. 
  Note that in our case, the solution is restricted to containing the edges of a predetermined forest,
  such that the solution is unique, 
  edge weights play no role, 
  and the challenge becomes to determine the feasible flow quickly.
  \item \emph{Forest-Function Evaluation (FFE).} 
  The remaining subtask is to assess if $\myfunction(\component)=1$,
  for each component $\component\subset \nodes$ in a set of disjoint, internally connected components \components---%
  i.e., to determine which of the components still need to be connected to others. 
  This is the only step that depends on \myfunction, 
  requiring an implementation of \myfunction matching the given model of computation.
\end{enumerate}
Note that \myfunction is an arbitrary proper function, 
so evaluating it can be arbitrarily hard.
Thus, our algorithm confines the hardness of the task that comes from the choice of \myfunction 
to $O(\epsilon^{-1}\log \nnodes)$ evaluations of $\myfunction(\component)$ for disjoint connected components~$\component\in\components$ (cf.~\Cref{relaxation:incremental-solution}).
In contrast, the other three subtasks can be solved by state-of-the-art algorithms from the literature in a black-box fashion. 

To illustrate the power of our result, 
we apply our machinery in three models of computation---%
\Congest, Parallel Random-Access Machine (PRAM), and \streaming (\streamingabbrv)---%
to three Constrained Forest Problems. 
Our results follow from \Cref{thm:metatheorem} with 
\begin{inparaenum}[(1)]
	\item the referenced results on aSSSP, MST, and \stepthreeabbrv, 
	\item model- and problem-specific subroutines for FFE, and 
	\item model-specific subroutines for book-keeping operations.  
\end{inparaenum}

In \Cref{tab:results}, we highlight the improvements over the state of the art 
that we achieve for our example problems
by instantiating our shell-decomposition algorithm in the \Congest model.
\begin{enumerate}[nosep,label=(\Roman*)]
	\item \emph{Steiner Forest (SF).} 
	From $\tildeO(\epsilon^{-1}(\shortestpathdiameter\ncomponents+\sqrt{\min\{\shortestpathdiameter\nterminals,\nnodes\}}))$ time\footnote{%
		Recall that \nnodes denotes the number of nodes, 
		\ncomponents the number of input components, 
		$\nterminals$ the number of terminals, 
		\hopdiameter the unweighted (hop) diameter, and 
		\shortestpathdiameter the shortest-path diameter. 
		Both \nterminals and \shortestpathdiameter can be up to $\Omega(\nnodes)$, 
		regardless of \ncomponents or \hopdiameter.
	} for a $(2+\epsilon)$-approximation and $\tildeO(\min\{\shortestpathdiameter,\sqrt{\nnodes}\}+\hopdiameter+\ncomponents)$ time for an $\BO(\log n)$-approximation to SF--IC obtained by \citet{lenzen2014steiner} to $(2+\epsilon)$-approximations in time
	\begin{inparaenum}[(1)]
	\item $\tildeO(\epsilon^{-3}(\sqrt{\nnodes}+\hopdiameter)+\epsilon^{-1}\ncomponents)$ for SF--IC,
	\item $\tildeO(\epsilon^{-3}(\sqrt{\nnodes}+\hopdiameter))$ for SF--SCR, and
	\item $\tildeO(\epsilon^{-3}(\sqrt{\nnodes}+\hopdiameter)+\epsilon^{-1}\min\{n^{2/3},k\})$ for SF--CIC.
\end{inparaenum}
	For SF--CR, we incur the same running-time overhead of $\BO(\nterminals)$ as \citet{lenzen2014steiner}, 
	matching the corresponding existential lower bound they showed.
	\item \emph{Point-to-Point Connection (PPC).} 
	We are not aware of prior work providing \Congest algorithms for the PPC problem. 
	Here, we obtain a $(2+\epsilon)$-approximation in $\tildeO(\epsilon^{-3}(\sqrt{\nnodes}+\hopdiameter))$~time.
	\item \emph{\facloclong (\facloc).} 
	We do not know of any existing \Congest algorithms for the \facloc problem,  
	and we realize a $(2+\epsilon)$-approximation in $\tildeO(\epsilon^{-3}(\sqrt{\nnodes}+\hopdiameter))$~time.
\end{enumerate}
We also derive algorithms for SF, PPC, and \facloc in the PRAM and Multi-Pass Streaming models, 
taking $\tildeO(\epsilon^{-3}\nedges)$ work and $\tildeO(\epsilon^{-3})$ depth in the PRAM model, 
as well as $\tildeO(\epsilon^{-3})$ passes and $\tildeO(\nnodes)$ space in the Multi-Pass Streaming model.
To the best of our knowledge, these algorithms are either the first ones to perform these tasks in their respective models or they cover more general classes of instances than the state of the art.
Notably, we obtain our diverse results with relative ease once the analysis of our model-agnostic shell-decomposition algorithm is in place---%
which contrasts with the challenges of directly designing specific algorithms for specific problems in specific models.

\paragraph{Taking Shortcuts.}
In the \Congest model, 
the $\sqrt{\nnodes}$ term in the complexity is due to the fact that, in general, 
it is not possible to solve \emph{Partwise Aggregation} (PA) in $\tildeo(\sqrt{\nnodes})$ rounds.
PA (see \cref{sec:preliminaries}) denotes the task of performing an aggregation or broadcast operation on each subset in a partition of \nodes that induces connected components \cite{ghaffari2021congestion}.
We can leverage PA, inter alia, 
to determine a leader and distribute its ID, 
or to find and make known a heaviest outgoing edge, 
for each component (a.k.a.\ \emph{part}) in parallel.
Such operations are key to efficient MST construction,
and any \pacomplexity-round solution to PA lets us solve MST in $\tildeO(\pacomplexity+\hopdiameter)=\tildeO(\pacomplexity)$ rounds \cite{ghaffari2016algorithms,ghaffari2021congestion}.

As MST computation is a CFP, 
it might not surprise that Partwise Aggregation can serve as a key subroutine for other CFPs in the \Congest model as well.
We show that the $\sqrt{\nnodes}$ term in the complexity of CFPs can be replaced by \pacomplexity.
Since this is already known for
$(1+\epsilon)$-approximate Set-Source Shortest-Path Forest \cite{rozhon2022paths},
Minimum Spanning Tree~\cite{ghaffari2016algorithms}, 
and $(1+\epsilon)$-approximate~Transship\-ment~\cite{rozhon2022paths} 
(which can be used to solve \stepthree),\footnote{%
	The results for approximate Set-Source Shortest-Path Forest and approximate Transshipment are conditional on the existence of fast $(\tildeO(1),\tildeO(1))$-cycle-cover algorithms; 
	otherwise, we incur an $\nnodes^{o(1)}$ overhead (see~\cref{tab:results}).
}  
again we need to show this for Forest-Function Evaluation only.
This is straightforward for PPC and \facloc, 
and simple algorithms evaluate the forest function for SF-IC and SF-CR in $\BO(\pacomplexity+\ncomponents)$ and $\BO(\pacomplexity+\nterminals)$ rounds, respectively.

The substantial literature on \emph{low-congestion shortcuts} provides a large array of results on solving Partwise Aggregation in time comparable to the \emph{shortcut quality} of the input graph, 
i.e., the best possible running time \pacomplexity for Partwise Aggregation~\cite{zuzic2022paths,rozhon2022paths,haeupler2021universally,haeupler2022expander,ghaffari2022universally,ghaffari2021congestion,haeupler2016shortcuts,haeupler2016nearoptimal,anagnostides2023almost}.
In particular, $\pacomplexity\in \tildeO(\hopdiameter)$ if the input graph does not contain a fixed $\tildeO(1)$-dense minor, 
without precomputations or further knowledge of \graph \cite{ghaffari2021congestion}.
Examples of graphs without $\tildeO(1)$-dense minors are planar graphs and, more generally, graphs of bounded genus.
Moreover, in \emph{Supported \Congest} (where $(\nodes,\edges)$ is known---or rather, can be preprocessed), 
$\pacomplexity\in \tildeO(\shortcutquality)$, i.e., within a polylogarithmic factor of the optimum.
Due to the modular structure of our results, 
any future results on Partwise Aggregation and low-congestion shortcuts will automatically improve the state of the art for CFPs in \Congest.

\paragraph{The Quest for Universal Optimality.}
In the \Congest model, 
it is known that even on a \emph{fixed} network topology, 
$\tildeOmega(\pacomplexity)$ rounds are required to obtain \emph{any} non-trivial approximation for a large class of problems.
This class includes MST, a special case of SF and ST, 
and shortest $s$-$t$-path~\cite{haeupler2021universally}, a special case of PPC.
Thus, this \emph{universal} lower bound applies to our example tasks as well. In particular, our results on PPC and \facloc have almost universally tight running times.

In contrast, the additive terms of \ncomponents and \nterminals in the running times of our algorithms for SF-IC and SF-CR, respectively, are only \emph{existentially} optimal~\cite{lenzen2014steiner}: 
The lower-bound graph---a double star---has shortcut quality $\BO(1)$,
but in a fully connected graph, 
it is trivial to evaluate \myfunction for all current components in $2$ rounds.
This motivates us to split Forest-Function Evaluation for Steiner Forest into two parts: 
\begin{inparaenum}[(1)]
	\item determining, for each input component $\nodes_i\subseteq\nodes$ or connection request $r\in\requests$, respectively, 
	whether the respective connectivity requirements are \emph{met}, and
	\item determining, 
	for each current component $\component \in \components$, 
	whether it contains a terminal whose connectivity requirements are \emph{not met}.
\end{inparaenum}
While the second part can be solved via standard Partwise Aggregation, 
the first part requires aggregating information within \ncomponents (or \nterminals) disjoint components that may be \emph{internally disconnected}. 
We call this task \emph{\graphtask on $p$ parts}, $\disjointagg(p)$, and propose the \emph{$p$-weave quality $\graphparameter(p)$} as a parameter to capture its universal complexity. 
Using this terminology, for constant $\epsilon$, our SF--IC and SF--CR algorithms take $\tildeO(\pacomplexity\nnodes^{o(1)} + \dacomplexity{\ncomponents})$ and $\tildeO(\pacomplexity\nnodes^{o(1)} + \dacomplexity{\nterminals})$ rounds in the \Congest model, respectively.
We conjecture that this becomes almost universally optimal if we parametrize SF--IC and SF--CR by restricting $\ncomponents \le p$ or $\nterminals \le p$, respectively, 
and obtain near-universally optimal solutions for $\disjointagg(p)$ with $\dacomplexity{p}\in \tildeO(\graphparameter(p))$. 
Note that we achieve the trivial bound $\dacomplexity{p}\in \BO(\hopdiameter+p)$ via pipelining over a BFS tree.

As an orthogonal approach, we explore the effect of the input specification on our SF results. 
The assumptions that
\begin{inparaenum}[(1)]
	\item pairs of terminals know that they need to be connected (SF--SCR), or
	\item the size of each input component is known to its constituent terminals (SF--CIC) 
\end{inparaenum}
both are plausible, 
and they change the basis of the applicable existential lower bounds from $2$-party \emph{set disjointness}~\cite{lenzen2014steiner} to $2$-party \emph{equality}, as detailed in \Cref{apx:lower}.

Assuming SF--SCR, we show how to achieve a running time of $\tildeO(\min\{\pacomplexity\nnodes^{o(1)},\sqrt{\nnodes}+\hopdiameter\})$ 
using efficient randomized $2$-party equality testing.
We sketch a solution assuming \emph{shared} randomness in our context---%
standard techniques achieve the same without shared randomness.
For each terminal-request pair $\{\othernode,\node\}$, 
we flip a fair independent coin and denote the result by $c_{\othernode,\node}\in \{0,1\}$.
Now each component \component aggregates $\sum_{\othernode\in \component}\sum_{\node\in \requests_v}c_{\othernode,\node}\bmod 2$.
Observe that if $\othernode,\node\in \component$ for request pair $\{\othernode,\node\}$, then for SF--SCR, 
it holds that $\node\in \requests_\othernode$ and $\othernode\in \requests_\node$.
Hence, if \component contains, for each request pair, 
either both terminals or none of the terminals, 
then the sum is guaranteed to be $0$ modulo $2$.
Otherwise, fix a request pair $\{\othernode,\node\}$ with $\othernode\in \component$ and $\node\in \nodes\setminus \component$.
After evaluating the sum up to coin $c_{\othernode,\node}$, it is either $0$ or $1$, 
and hence, 
by independence of $c_{\othernode,\node}$, 
with probability $\nicefrac{1}{2}$, the sum is $1$ modulo $2$.
Therefore, performing the process $\BO(\log \nnodes)$ times in parallel, we can distinguish between $\myfunction(\component) = 0$ and $\myfunction(\component) = 1$ with high probability.
This computation can be performed by a single aggregation, 
where the aggregation operator $\bigoplus$ is given by bit-wise addition modulo $2$.

Our strategy for SF--CIC is similar, but results in a much weaker bound of $\tildeO(\nnodes^{\nicefrac{2}{3}}+\hopdiameter)$;
our main point here is to demonstrate that the $\tildeOmega(\ncomponents)$ bound can be beaten.
We distinguish three cases.
\begin{inparaenum}[(1)]
  \item Components \component of size at most $\nnodes^{\nicefrac{2}{3}}$ are spanned by a rooted tree of size $\nnodes^{\nicefrac{2}{3}}$. 
  Here, we can aggregate the terminal counts, for all input components, within \component at \component's root 
  in $\BO(\nnodes^{\nicefrac{2}{3}})$ rounds to determine activity status.
  \item For each input component of size at least $\nnodes^{\nicefrac{1}{3}}$, 
  we globally aggregate if there are two distinct component IDs with a terminal from that input component. 
  This requires one aggregation for each such input component, 
  all of which can be completed within $\BO(\nnodes^{\nicefrac{2}{3}}+\hopdiameter)$ rounds via a BFS tree, 
  as there can be at most $\nnodes^{\nicefrac{2}{3}}$ input components of this size.
  \item For input components of size $s<\nnodes^{\nicefrac{1}{3}}$, each component \component of size larger than $\nnodes^{\nicefrac{2}{3}}$ uses the same strategy as for SF--SCR. 
  However, only input components of the same size can be handled in a single aggregation, as the summation is now modulo $s$. 
  Hence, $\tildeO(\nnodes^{\nicefrac{1}{3}})$ aggregations by at most $\nnodes^{\nicefrac{1}{3}}$ components of size larger than $\nnodes^{\nicefrac{2}{3}}$ are required, 
  which again can be performed within $\tildeO(\nnodes^{\nicefrac{2}{3}}+\hopdiameter)$ rounds.
\end{inparaenum}

%% file: text/conclusion.tex
\section{Discussion}
\label{sec:conclusion}

In this work, we presented a general model-agnostic framework for the $(2+\varepsilon)$-approximation of Constrained Forest Problems (CFPs) and demonstrated its utility on three NP-hard CFPs in three models of computation.
We conclude with a number of open questions---beyond applying our framework to other graph problems and computational models---in increasing order of generality:
\begin{compactenum}[(Q1)]
	\item 
	\textbf{Can the running time of SF--CIC in} \Congest \textbf{be improved to nearly $\pacomplexity$?}
  \item Many of our results are near-universally optimal in \Congest, even for randomized algorithms that succeed deterministically, i.e., Las Vegas algorithms---%
  but our algorithms for SF-SCR and SF-CIC are Monte Carlo. 
  Due to existential lower bounds based on the communication complexity of $2$-party equality, this is required to (always) achieve small running times w.h.p. \\
  \textbf{Is $\tildeOmega(\shortcutquality)$ also a lower bound for Las Vegas} \Congest \textbf{algorithms?}
  \item \textbf{How hard is Disjoint Aggregation, i.e., how can we characterize $\dacomplexity{p}$ and~$\graphparameter(p)$?}
	\item The \facloc problem minimizes the sum of opening costs and forest costs, disregarding~the (typically distance-based) \emph{connection costs} considered in other facility-location-type problems.\\
	\textbf{Can our approach to \facloc be generalized to problems including connection costs?}
  \item As we reduce most of our tasks to few PA instances, 
  in \Congest, $\pacomplexity \approx \shortcutquality$ rounds are both necessary and sufficient to achieve near-universal optimality. 
  Since PA and DA can be solved in $\BO(\nnodes)$ work and ${\BO(\log \nnodes) = \tildeO(1)}$ depth in PRAM, 
  PA-based algorithms also yield good solutions in PRAM. 
  While PA and DA can be solved in $\tildeO(\nnodes)$ memory and two passes in the streaming model, 
  it is unclear if this yields optimal results (unless we insist that the output needs to be held in memory),
  as we are mostly limited to existential $\tildeOmega(\sqrt{\nnodes})$ lower bounds.\\
  \textbf{Are there $\tilde{o}(\nnodes)$-memory streaming algorithms with $\tildeO(1)$ passes?
  	Or, better yet: Is there an analog to universal optimality in the Multi-Pass Streaming model?}\\
  	Note that PA and $\disjointagg(p)$ seem inadequate as parameters here, 
  	as each part will require \emph{some} memory for a few-pass implementation, but there may be $\Omega(\nnodes)$ parts.
   \item 
   \textbf{Does our approach generalize to Constrained Forest Problems on hypergraphs?}
  \item 
  Beyond proper functions, the primal-dual method has proven useful when working with \emph{uncrossing} functions \cite{goemans1996primal}. 
  One of the main features of optimization problems with uncrossing functions is that they are guaranteed to feature an optimal dual solution that is \emph{laminar}. 
  \textbf{Can our approach be extended to uncrossing or other non-proper functions?} 
\end{compactenum}

%% file: text/appendix-notation.tex
\section{Notation}
\label{apx:notation}

In \cref{tab:notation}, we summarize the main notation used in this work.

\vspace*{1em}

\begin{table}[h]
	\centering
	\caption{%
		Main notation used in this work. 
		All notation involving nodes and edges makes implicit reference to the graph \graph.
	}
	\label{tab:notation}
	\input{text/table-notation}
\end{table}

%% file: text/table-notation.tex
\footnotesize
\begin{tabular}{r@{\hspace*{0.66em}}c@{\hspace*{0.66em}}l@{\hspace*{0.95em}}l}
	\toprule
	Symbol&&Definition&Meaning\\
	\midrule
	$\cardinality{\thesubset}$&&&Cardinality of a set \thesubset\\
	$2^\thesubset$&$=$&$\{X\mid X\subseteq \thesubset\}$&Set of all subsets (i.e., \emph{power set}) of \thesubset\\
	$\binom{\thesubset}{k}$&$=$&$\{X\subseteq\thesubset\mid\cardinality{X}=k\}$&Set of $k$-element subsets of \thesubset\\
	$[k],[k]_0$&$=$&$\{i\mid i\in\naturalsnozero, i \leq k\}, \{i\mid i\in\naturals, i \leq k\}$&Set of nonnegative resp. positive integers no larger than $k$\\
	\midrule
	\graph&$=$&$(\nodes,\edges)$&Graph with node set \nodes and edge set \edges\\
	\nnodes&$=$&$\cardinality{\nodes}$&Number of nodes\\
	\nedges&$=$&$\cardinality{\edges}$&Number of edges\\
	$\cost(\edge)$&$\in$&$\naturalsnozero = \{1,2,\ldots\}$&Cost of edge \edge\\
	$\cost(\edgesubset)$&$=$&$\sum_{\edge\in\edgesubset}\cost(\edge)$&Cost of edge subset $\edgesubset\subseteq\edges$\\
	\midrule
	$\somepath(\othernode,\node)$&$=$&$(\edge_1,\ldots, \edge_\ell)$ s.t. $\exists(\node_1,\dots,\node_{\ell+1})$: &$\ell$-hop path between \othernode and \node (distinct nodes and distinct edges)\\
	&&$\edge_i = \{\node_i,\node_{i+1}\} \in  \edges$ $\forall i\in[\ell]$, &\\
	&&$\othernode = \node_1$, $\node = \node_{\ell+1}$&\\
	$\hopdistance(\othernode,\node)$&$=$&$\min\{i\mid \exists~\somepath(\othernode,\node)~\text{with}~\cardinality{\somepath(\othernode,\node)} = i\}$&Hop (= unweighted) distance between \othernode and \node\\
	\hopdiameter&$=$&$\max\{\hopdistance(\othernode,\node)\mid\{\othernode,\node\}\in\binom{\nodes}{2}\}$&Hop diameter of \graph\\
	$\distance(\othernode,\node)$&$=$&$\min\{i \mid \exists~ \somepath(\othernode,\node)~\text{with}~ \cost(\somepath(\othernode,\node)\}$&Shortest-path (= weighted) distance between \othernode and \node\\
	$\shortestpath(\othernode,\node)$&$=$&$\somepath(\othernode,\node)$ with $\cost(\somepath(\othernode,\node)) = \distance(\othernode,\node)$&Shortest path between \othernode and \node\\
	\shortestpathdiameter&$=$&$\max\{\min \{\cardinality{\shortestpath(\othernode,\node)} \mid \{\othernode,\node\}\in\binom{\nodes}{2}\}$&Shortest-path diameter of \graph\\
	$\cut(\thesubset)$&$=$&$\{\edge\in\edges\mid \cardinality{\edge\cap\thesubset} = 1\}$&Set of edges with exactly one endpoint in $\thesubset\subset \nodes$\\
	\midrule
	\terminals&$=$&$\{\node\in\nodes\mid\myfunction(\{\node\}) = 1\}$&Terminals\\
	$\nterminals$&$=$&$\cardinality{\terminals}$&Number of terminals\\
	\ncomponents&&&Number of input components in SF--IC\\
	\bottomrule
\end{tabular}

%% file: text/appendix-algorithms.tex
\section{Model-Agnostic Algorithm}
\label{apx:centralized-algorithm}

In this section, we develop our \emph{modified} centralized \GW algorithm 
and establish the correctness of our \emph{model-agnostic} implementation.
For completeness, we state the original \emph{centralized} \GW algorithm as \cref{alg:gw-original}. 
This algorithm yields a $(2-\nicefrac{2}{\nterminals})$-approximation, as proved by \citet{goemans1995approximation}. 

\begin{theorem}[Original \GW Algorithm \cite{goemans1995approximation}]
	\label{thm:apx-original}
	For a graph $\graph=(\nodes,\edges)$, edge costs ${\cost: \edges\rightarrow \naturalsnozero}$, and  
	proper forest function \myfunction, 
	the original \GW algorithm (\cref{alg:gw-original}) yields a $(2-\nicefrac{2}{\nterminals})$-approxima\-tion to the optimal solution.
\end{theorem}

\subsection[Modified Centralized (2+eps)-Approximation]{Modified Centralized $(2+\epsilon)$-Approximation}

The original centralized \GW algorithm assumes 
the presence of a final filtering step to eliminate unneeded edges from the solution, 
as well as 
exact distances and 
immediate forest-function evaluation. 
Since these assumptions make the algorithm hard to implement efficiently in different computational models, 
we modify it to work with 
\begin{inparaenum}[(1)]
	\item the \emph{absence} of a final filtering step (i.e., with incremental solution-set construction), 
	\item $(1+\distanceeps)$-\emph{approximate} distances, 
	and 
	\item $(1+\overgrowth)$-\emph{deferred} forest-function evaluation. 
\end{inparaenum}
Stating the result, 
which we call the \emph{shell-decomposition algorithm}, 
as \cref{alg:gw-clean}, 
and illustrating its basic operations in \Cref{fig:gw-clean-illustration},
we now prove the following theorem.

\begin{restatable}{theorem}{gwcleanapx}\label{thm:gwcleanapx}
	For a graph $\graph = (\nodes,\edges)$, edge costs $\cost\colon\edges\to\naturalsnozero$, 
	and proper function \myfunction, 
	the modified \GW algorithm with incremental solution-set construction, 
	$(1+\distanceeps)$-approximate distance computations, 
	and $(1+\overgrowth)$-deferred forest-function evaluation (\cref{alg:gw-clean}) yields a ${(2-\nicefrac{2}{\nterminals})(1+\distanceeps)(1+\overgrowth)^2}$-approximation to the optimal solution, 
	i.e., a $(2+\epsilon)$-approximation for $\distanceeps,\overgrowth\le \nicefrac{\epsilon}{4}\le \nicefrac{1}{4}$.
\end{restatable}

\begin{figure}[b!]
	\centering
	\includegraphics[width=0.9\linewidth]{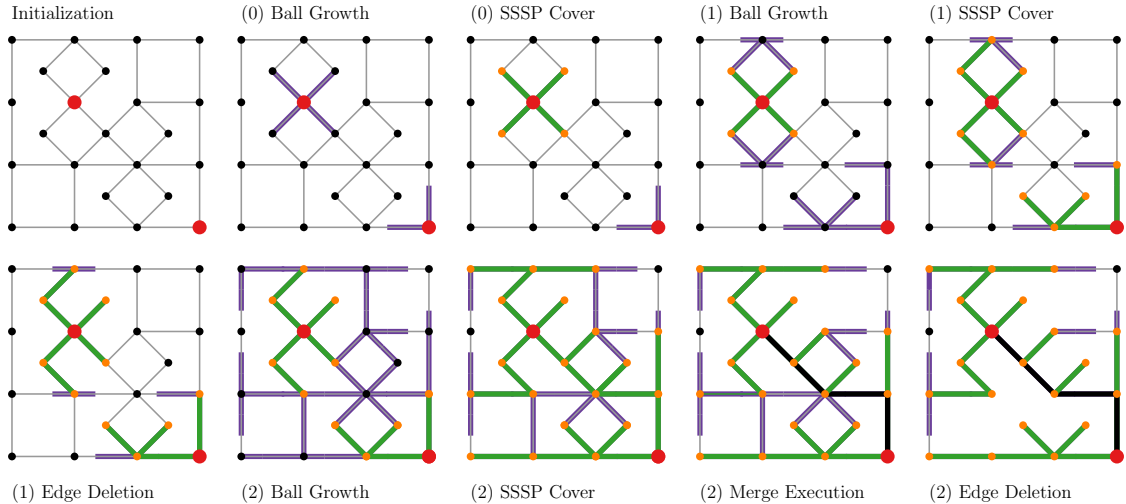}
	\caption{%
		Operation of our shell-decomposition algorithm (\cref{alg:gw-clean}) on an $s$-$t$-shortest-path instance seeking to connect the red nodes, 
		starting with $\radius_0 = \nicefrac{1}{2}$ (cf.~Line~\ref{line:init-radius}), and 
		working over phases $0$, $1$, and $2$.
		Panels are labeled with their phase number and the illustrated step of \cref{alg:gw-clean}.
		Nodes absorbed by the SSSP forest are drawn in orange, 
		edge-cost reduction is indicated in purple, 
		edges selected into the SSSP forest are marked in green, 
		and edges selected into the solution are marked in black. 
		Distance approximations and deferred forest-function evaluation are not shown.
	}\label{fig:gw-clean-illustration}
\end{figure}

\clearpage

\input{text/gw-original}

\input{text/gw-graph-clean}

\clearpage

In the following, denote by $i=0,1,\ldots$ the iterations of the while loop in \cref{alg:gw-clean}, calling each iteration a $\emph{merge phase}$, 
and for each phase
\begin{compactitem}[---]
\item by $\radius_i:=(1+\overgrowth)^i\cdot\frac{\overgrowth}{4}$ the radius $r$ at the \emph{beginning} of phase~$i$ (with $\radius_{-1}:=(1+\overgrowth)^{-1}\cdot\frac{\overgrowth}{4}$),
\item by $U_i$, $\reducedcost_i$ $\forest_i$, $\forest'_i$, $\augmentation_i$, $\edges_i$, and $\terminals^1_i$ the values of the respective variables at the \emph{end} of phase~$i$, and 
\item by $\nactive_i\coloneq\cardinality{\terminals^1_i}$ the number of active components at the \emph{end} of the phase~$i$ (with $\nactive_{-1}\coloneq\nterminals$).
\end{compactitem}

To prepare our proof of \cref{thm:gwcleanapx}, 
we first make a number of basic observations that can be readily verified from the pseudocode of \Cref{alg:gw-clean}.
\begin{lemma}\label{lem:gwcleanbasic}
For all $i\in \naturals$ such that \Cref{alg:gw-clean} does not terminate at the end of phase $i$, we have
\begin{compactenum}[(i)]
  \item $\radius_i=(1+\overgrowth)r_{i-1}$,
  \item $\reducedcost_{i+1}(e)\le \reducedcost_i(e)\le \cost(e)$ for each $e\in \edges_i$,
  \item $\ballunion_i$ is spanned by $\forest'_i\cup \augmentation_i$,
  \item $\ballunion_i\subseteq \ballunion_{i+1}$,
  \item $\forest_i'\cup A_i\subseteq \forest_{i+1}'$,
  \item $\forest_i\subseteq \forest_{i+1}$,
  \item $\forest_i\cup \forest_i'=\forest_i'\cup \augmentation_i'$ and
  \item $(\nodes,\edges_i)$ is connected and the weighted diameter with respect to reduced costs \reducedcost is decreasing.
\end{compactenum}
\end{lemma}
\begin{proof}
The first and the second statement hold by construction due to Lines~\ref{line:increase-radius} and~\ref{line:reduced-costs}, respectively.
For the third statement, observe that $\forest'_i$ connects each node in $\ballunion_i$ to some node in $\terminals^1_i$, 
and the set $\mergeedges_i$ includes all edges that are both contained in $\ballunion_i$ (i.e., whose reduced cost has become $0$) and connect different connectivity components with respect to $\forest_i'$ of $\ballunion_i$. 
Thus, the choice of $\augmentation_i$ ensures that $\ballunion_i$ is spanned by $\forest'_i\cup \augmentation_i$.
This readily implies the forth statement, since $\forest'_i\cup \augmentation_i\subseteq \edges_i$, 
as well as the fifth, since all other edges \edge of reduced cost $\reducedcost(\edge) = 0$ are removed to obtain $\edges_i$, forcing any approximate SSSP solution to reselect the edges from $\forest'_i\cup\augmentation_i$ into $\forest'_{i+1}$.
The sixth statement holds because the algorithm only adds edges to $\forest$,
and the seventh statement holds by induction, 
using that $\forest_i=\forest_{i-1}\cup \augmentation_i\cup X$ for some set $X\subseteq \forest_i'$.
Finally, for the eighth statement, 
observe that since $\ballunion_i$ is spanned by $\forest'_i\cup \augmentation_i$, 
whose reduced cost is $0$, and any edges not fully contained in $\ballunion_i$ remain in $\edges_i$, 
the shortest path between any pair of nodes with respect to \reducedcost only becomes shorter: 
For any edge contained in $U_i$, there is now a path of reduced weight $0$ between its endpoints, while any edge \edge with reduced cost $\reducedcost_i(\edge) > 0$ is still present and retains at most its original cost $\cost(\edge)$.
\end{proof}

From the above initial observations, we can readily infer that \cref{alg:gw-clean} always terminates, 
and for reasonable choices of \distanceeps and \overgrowth, it does so after a small number of loop iterations.

\begin{restatable}{lemma}{gwcleantermination}\label{lem:gwcleantermination}
	For $\distanceeps,\overgrowth\in\nnodes^{-\BO(1)}$, 
	\cref{alg:gw-clean} terminates within $\BO(\frac{\log n}{\overgrowth})$ iterations of its while loop.
\end{restatable}
\begin{proof}
	Since each of the operations within the loop terminate (assuming correct subroutines), it suffices to show the claimed bound on the number of loop iterations. 
	To exit the while-loop, we must have $\myfunction(\component) = 0$ for each $\component\in\components$, 
	where \components is the set of connected components induced by our current edge set \forest.
	Recall that edge weights and hence the weighted diameter of the graph are polynomially bounded in \nnodes.
	As $\overgrowth,\distanceeps\in n^{-\BO(1)}$, there is an $j\in \BO(\frac{\log n}{\overgrowth})$ for which $r_j=(1+\overgrowth)^{j}\cdot\frac{\overgrowth}{4}$ exceeds the weighted diameter of \graph times $1+\distanceeps$.
	By \Cref{lem:gwcleanbasic}~(viii), 
	this entails that if we reach this phase, 
	then $\ballunion_j$ contains the entire connected graph $(\nodes,\edges_j)$.
	By \Cref{lem:gwcleanbasic}~(iii), $\ballunion_j$ (and thus $\nodes$) is spanned by $\forest'_j\cup \augmentation_j$.
	The merge operation in Line~\ref{line:perform-merges} hence ensures that all terminals in $\terminals^1_{j-1}$ are connected by $\forest_j$.
	Denote by $\component$ the connectivity component of $(\nodes,\forest_j)$ containing the nodes in $\terminals^1_{j-1}$.
	For each terminal $\tau\notin \terminals^1_{j-1}\setminus \component$, $\tau$ lies in a connected component $\component'$ of $(\nodes,\forest_{j-1})$ satisfying that $\myfunction(\component')=0$.
	This entails that we can partition $\nodes\setminus \component$ into two types of sets:
\begin{inparaenum}[(i)]
\item components $\component'$ of $(\nodes,\forest_{j-1})$ satisfying that $\myfunction(\component')=0$, and
\item non-terminals $\node \in \nodes\setminus \terminals$, which satisfy $\myfunction(\{\node\})=0$ by definition.
\end{inparaenum}
By disjointness of the proper forest function $\myfunction$, it follows that $\myfunction(\nodes\setminus \components)=0$, and by symmetry it follows that $\myfunction(\component)=0$.

Finally, by \Cref{lem:gwcleanbasic}~(vi), we have $\forest_{j-1}\subseteq \forest_{j}$.
Therefore, each connectivity component $\component'$ of $(\nodes,\forest_{j})$ other than $\component$ decomposes into connectivity components of $(\nodes,\forest_{j-1})$, 
satisfying that $\myfunction(\component')=0$, 
and non-terminals, which again by disjointness implies that $\myfunction(\component')=0$.
We conclude that \Cref{alg:gw-clean} terminates at the latest by the end of phase $j\in \BO(\frac{\log n}{\overgrowth})$, concluding the proof.
\end{proof}

Next, we ascertain that \cref{alg:gw-clean} outputs a primal feasible solution.

\begin{restatable}{lemma}{gwcleanprimalfeasible}\label{lem:gwcleanprimalfeasible}
	The edge set \forest output by \cref{alg:gw-clean} is primal feasible.
\end{restatable}
\begin{proof}
	When \cref{alg:gw-clean} terminates, we must have exited its while-loop, 
	implying that we have $\myfunction(\component) = 0$ for each $\component\in\components$, 
	where \components is the set of connected components induced by our output set \forest.
	
	Consider any set $S\subseteq V$.
	If $S$ non-trivially intersects a component $\component\in \components$, i.e., $\emptyset\neq S\cap \component\neq \component$, then there is an edge of $\forest$ in the cut defined by $S$, i.e., $|\delta(S)\cap \forest|\ge 1\ge \myfunction(S)$.
	Otherwise, $S=\bigcup_{\component\in \components_S}\component$ for some $\components_S\subseteq \components$.
	As $\myfunction$ satisfies disjointness and $\myfunction(\component)=0$ for each $\component\in \components_S$, it follows that $|\delta(S)\cap \forest|\ge 0=\myfunction(S)$ in this case as well.
	Thus, we conclude that \forest is primal feasible.
\end{proof}

To prove our approximation ratio, 
we also need to ensure that the value \lowerbound output by \cref{alg:gw-clean} is associated with a feasible dual solution.

\begin{restatable}{lemma}{gwcleandualfeasible}\label{lem:gwcleandualfeasible}
	The value \lowerbound output by \cref{alg:gw-clean} is the cost of a feasible dual solution. 
\end{restatable}
\begin{proof}
	Recall that $\radius = \radius_i=(1+\overgrowth)^i\frac{\overgrowth}{4}$ in phase~$i$, 
	and denote by $j$ the index of the final phase
	Abbreviating $\Delta_{i} \coloneq (1+\distanceeps)^{-1}\radius_{i}$, 
	we can write the value \lowerbound as
	\begin{align*}
		\lowerbound 
		= 
		(1+\distanceeps)^{-1}
		\sum_{i=0}^{j} (1+\overgrowth)^{i}\frac{\overgrowth}{4}\nactive_{i}
		= 
		(1+\distanceeps)^{-1}
		\sum_{i=0}^{j} \radius_{i}\nactive_{i}
		=
		\sum_{i=0}^{j}\Delta_{i}\nactive_{i}
		\;.
	\end{align*}
	
	Denote the cost reduction of an edge achieved in phase~$i$ 
	as  $\costreduction_{i}(\edge) \coloneq \cost_{\ballunion_{i}}(\edge) = \reducedcost_{i-1}(\edge) - \reducedcost_{i}(\edge)$, 
	with $\reducedcost_{-1}(\edge) \coloneq \cost(\edge)$,
	and the amount that $\yvar$ is increased in phase~$i$ 
	for some set $\thesubset\subset\nodes$ as $\yvar_i(\thesubset)$.
	Now observe that to construct a feasible dual solution of value \lowerbound, 
	it suffices to, 
	in each phase~$i\in[j]_0$,
	for each component \component surviving phase $i$,
	increase the dual variables associated with the subsets $\thesubset\subset\component$ in sum by $\Delta_{i}$, 
	while ensuring that for each $\edge\in\edges$, we maintain
	\begin{align*}
		\sum_{\thesubset\colon\cut(\thesubset)\ni\edge}\yvar_i(\thesubset) \leq \costreduction_{i}(\edge)\;.
	\end{align*} 

 	 Since \myfunction is proper, 
 	we know that for each component \component surviving phase~$i$,
 	there exists at least one component $\component'$ active in phase~$i$
 	such that $\component'\subseteq \component$.
 	Therefore, 
 	to allocate $\Delta_{i}$ to dual variables associated with \component, 
 	we can trace the construction of $\component$ from the set of components $\components' \coloneq \{\component'\subseteq\component \mid \myfunction(\component') = 1,$ 
 	$\component'$ is a connected component of $(\nodes,\forest_{i-1})\}$ (where $\forest_{-1} \coloneq \emptyset$) as follows.
 	Starting with $\components'$ as it resulted from phase $i-1$, 
 	we increase the radii of all components $\component'\in\components'$ at rate $\rho$ 
 	and the \yvar-variables of $\component'\in\components'$ at rate $\frac{\rho}{1+\distanceeps}$.
 	Thus, we gradually construct $\ballunion_{i}$ and reduce the costs of all edges in the affected cuts (i.e., increase $\cost_{\ballunion_{i}} = \costreduction_{i}(\edge)$)
 	until the first edge achieves $\reducedcost_i(\edge) = 0$ (or the phase ends). 
 	This happens \emph{at the latest} when the first dual constraint in phase~$i$ becomes tight---%
 	it can happen earlier as the edge cost reductions associated with our radius increases are based on $(1+\distanceeps)$-approximate distances. 
 	When an edge achieves $\reducedcost_i(\edge) = 0$, 
 	we update $\forest$ and $\components'$ to ensure that both endpoints of \edge lie in the same component,
 	and we iterate the process described~above.
 	
 	Since \component survived phase~$i$, this process does not add an edge to \forest that lies in the cut $(\component,\nodes\setminus\component)$ until we have increased the radius by~$\radius_{i}$. 
 	We claim that at no point in this process, $\components'$ becomes empty.
 	Assuming otherwise, we would have that $\component = \bigcup_{\component'\in\components'}\component'$ with $\myfunction(\component')=0$ at the respective point of the process, implying the contradiction that $\myfunction(\component)=0$ by disjointness of $\myfunction$.
 	Accordingly, the total increase of \yvar-variables that we attribute to \component during phase~$i$ is at least $\Delta_i$, as desired.
 	Moreover, our direct coupling of the \yvar-variable increases with edge-cost reductions further ensures 
 	that the \yvar-variables relevant for any individual edge \edge increase by at most $\costreduction_{i}(\edge)$, 
 	i.e., its cost reduction in phase~$i$, 
 	guaranteeing feasibility. 
 	
 	Summing over the $\nactive_i$ components surviving phase~$i$, we increase the dual variables by at least $\Delta_{i}\nactive_{i}$ per merge phase, concluding the proof. 
\end{proof}

Finally, we bound the cost of \forest in relation to \lowerbound.
\begin{lemma}\label{lem:gwcleanprimalcost}
Denoting by $j$ the phase after which \Cref{alg:gw-clean} terminates, it holds that
\begin{equation*}
		\cost(\forest)\le \left(2-\frac{2}{\nterminals}\right)(1+\overgrowth)^2\sum_{i=0}^{j-1}\radius_i\nactive_i\;.
	\end{equation*}
\end{lemma}
\begin{proof}
Recall that merge phase~$i\geq 0$ is the phase in which $\radius =\radius_i= (1+\overgrowth)^i\cdot\frac{\overgrowth}{4}$, 
	and $\nactive_i$ is the number of active components \emph{surviving} phase~$i$.
	Moreover, note that all edges \edge added to \forest satisfy that $\reducedcost(\edge) = 0$ on termination, 
	and that the cost reduction for each edge in phase~$i$ is at most $2\radius_i$, as each endpoint of the edge is contained in at most one tree of $\forest'$.
	Hence, the cost of each edge in \forest can be amortized over the phases~$i$ in which its cost is reduced by the algorithm, i.e., which it starts with $\reducedcost_{i-1}(\edge)>0$, and in which it intersects $\ballunion_i$.
	Accordingly, \forest decomposes into a set of nested \emph{shells} $\ballunion_i\setminus\ballunion_{i-1}$ (with $\ballunion_{-1}:=\emptyset$), which the algorithm iteratively and implicitly constructs around active components. 
	This \emph{shell-decomposition argument} is illustrated in \Cref{fig:shelling}.
	
	\begin{figure}[t]
		\centering
		\includegraphics[width=0.8\linewidth]{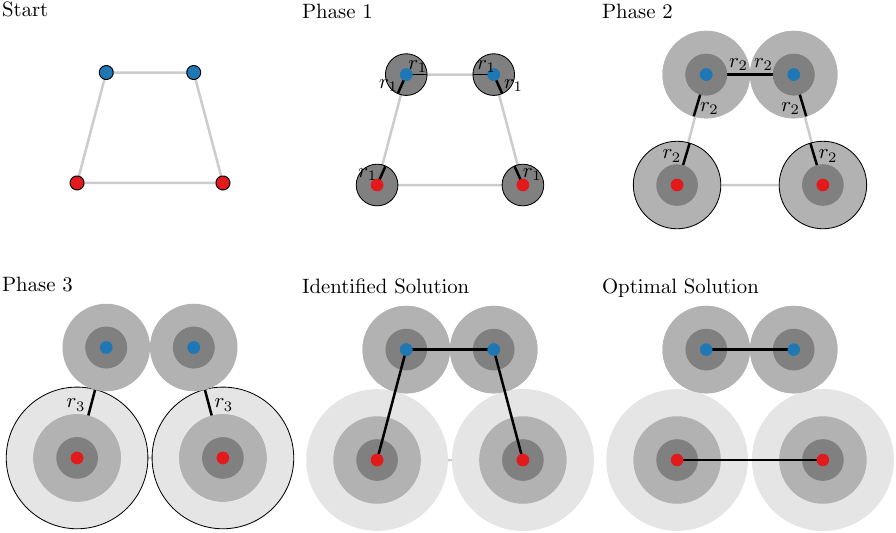}
		\caption{Illustration of our \emph{shell-decomposition argument} on a small instance of Steiner Forest (SF--IC). 
			Gray lines indicate original edges, black lines indicate (parts of) selected edges, black circle linings indicate active components, and node colors indicate input components. 
			The illustrations provided by \citet{goemans1995approximation} (Figs.~2--5) are stylistically similar, 
			but our visualization clarifies the phase-wise charging argument underlying our shell-decomposition algorithm. 
		}\label{fig:shelling}
	\end{figure}
	
	Since in phase~$i$, 
	the remaining edges $\forest\setminus \forest_{i-1}$ (where $\forest_{-1}\coloneq\emptyset$) to be added to \forest form a forest with $\nactive_i$ active components as nodes,  
	and the average degree of such a forest is at most $2-\nicefrac{2}{\nactive_i}$,
	we can bound the cost of the forest \forest output by \cref{alg:gw-clean}~from above as
	\begin{align*}
		\cost(\forest) 
		\leq 
		\sum_{i=0}^{j}\radius_i\nactive_{i-1}\left(2-\frac{2}{\nactive_{i-1}}\right)
		&\leq 
		\left(2-\frac{2}{\nterminals}\right)(1+\overgrowth)\sum_{i=0}^{j}\radius_{i-1}\nactive_{i-1}
		=\left(2-\frac{2}{\nterminals}\right)(1+\overgrowth)\sum_{i=-1}^{j-1}\radius_i\nactive_i\;,
	\end{align*}
	where $\nactive_{-1} \coloneq \nterminals$.
	Since \forest is a feasible primal solution by \Cref{lem:gwcleanprimalfeasible}, each terminal must be incident with at least one edge from \forest, 
	and the minimum edge weight is at least $1$, we have that $\cost(\forest)\ge \nicefrac{\nterminals}{2}$.
	On the other hand, $\nactive_{-1}\radius_{-1}= \frac{\overgrowth t}{4(1+\overgrowth)}$, yielding that
	\begin{equation*}
		\cost(\forest)\le (1+\overgrowth)\cost(\forest)-\frac{\overgrowth t}{2}=(1+\overgrowth)(\cost(\forest)-2\nactive_{-1}\radius_{-1})\le \left(2-\frac{2}{\nterminals}\right)(1+\overgrowth)^2\sum_{i=0}^{j-1}\radius_i\nactive_i\;.\qedhere
	\end{equation*}
\end{proof}

Using the above lemmas, we can now prove \cref{thm:gwcleanapx}.

\begin{proof}[Proof of \cref{thm:gwcleanapx}]
	Assume that \Cref{alg:gw-clean} terminates at the end of phase~$j$, i.e., $\nactive_{j} = 0$;
	by \Cref{lem:gwcleantermination}, such a phase exists.
	Observe that the value \lowerbound output by \Cref{alg:gw-clean} then satisfies
	\begin{align*}
		\lowerbound 
		= 
		(1+\distanceeps)^{-1}
		\sum_{i=0}^{j} \radius_{i}\nactive_{i}
		= 
		(1+\distanceeps)^{-1}
		\sum_{i=0}^{j-1} \radius_i\nactive_i\;.
	\end{align*}
	By \Cref{lem:gwcleanprimalfeasible}, \forest is a feasible primal solution. 
	As \lowerbound is the cost of a feasible dual solution by \Cref{lem:gwcleandualfeasible}, 
	using \Cref{lem:gwcleanprimalcost}, 
	we obtain the desired approximation guarantee as 
\begin{equation*}
	\frac{\cost(\forest)}{\lowerbound} 
	\leq 
	\frac{
		\left(2-\nicefrac{2}{\nterminals}\right)(1+\overgrowth)^2\sum_{i=0}^{j-1}\radius_i\nactive_i
	}{
	(1+\distanceeps)^{-1}
	\sum_{i=0}^{j-1} \radius_i\nactive_i
	}
	= \left(2-\frac{2}{\nterminals}\right)(1+\distanceeps)(1+\overgrowth)^2
	\;.\qedhere
\end{equation*}
\end{proof}

\subsection{Specification of our Model-Agnostic Algorithm}
\label{spec:meta-algo}

The task of approximating Constrained Forest Problems to within a factor of $2 + \epsilon$ in any specific model 
now boils down to implementing, or more precisely simulating, \Cref{alg:gw-clean}.
Here and in what follows, we do not discuss the computation of \lowerbound, which is straightforward to add if so desired.
By \Cref{lem:gwcleantermination}, 
we can divide the computation into $\BO(\frac{\log\nnodes}{\overgrowth})$ phases, 
where, starting at phase~$0$, 
phase~$i$ grows components by radius $\radius_i = (1 + \overgrowth)^i\cdot\frac{\overgrowth}{4}$. 
In our model-agnostic algorithm, 
each phase then consists of six main steps, 
as summarized in \Cref{fig:overview}.
In each phase, we tackle six abstract problems, 
which correspond to the six building blocks of our algorithm.

\begin{figure}[t]
	\centering
	\includegraphics[width=0.9\linewidth]{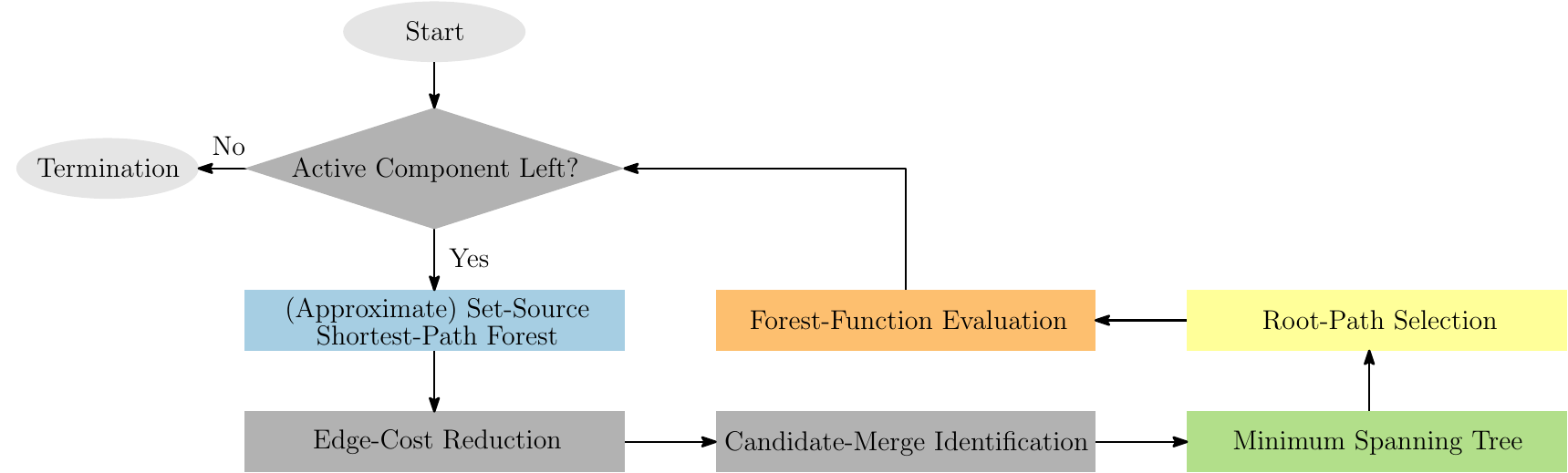}
	\caption{%
		Overview of our model-agnostic shell-decomposition algorithm for approximating Constrained Forest Problems.
		Main tasks are colored;
	simple book-keeping operations are shaded in~gray.
	}\label{fig:overview}
\end{figure}

\subsubsection{Problems Used as Building Blocks}
\label{sec:algorithm:modelagnostic:problems}

Our algorithm is composed of subroutines handling the following problems.
The input and output representation are model-specific, so we do not state them here.

\begin{problem}[$\alpha$-approximate Set-Source Shortest-Path Forest]
	\label{def:aSSSP}
	Given a connected graph $\graph = (\nodes,\edges)$ with edge costs $\cost\colon \edges \rightarrow \naturals$ 
	and a set of sources $\thesubset \subseteq \nodes$, 
	compute a forest $\forest'$ spanning \graph 
	such that for all nodes $\node \in \nodes$, $\distance_{\forest'}(\thesubset, \node) \leq \alpha \distance(\thesubset, \node)$, 
	where $\distance(\thesubset, \node) =  \min_{\othernode \in \thesubset}\{\distance(\othernode,\node)\}$, and $\distance_{\forest'}(\othernode,\node)$ is the weighted distance between \othernode and \node in \forest.
\end{problem}

Note that in phase $i$, we can confine ourselves to computing an $\alpha$-approximate set-source shortest-path forest \emph{up to distance $\radius_i$} (see~\cref{alg:gw-clean}, Lines~\ref{line:growth-balls}--\ref{line:growth-sptree}).

\begin{problem}[Candidate-Merge Identification]
	\label{def:cmi}
	Given a graph $\graph = (\nodes,\edges)$, edge costs ${\cost\colon \edges \rightarrow \naturals}$,
	and a rooted forest $\forest'$ with a subset of its trees marked such that each node \node in a marked tree knows that its tree is marked as well as the identity of its root $\terminal_\node$, 
	identify all edges $\edge = \{\othernode,\node\}$ that are in distinct marked trees and satisfy $\cost(\edge)=0$.
\end{problem}

\begin{problem}[Minimum Spanning Tree]
	\label{def:mst}
	Given a graph $\graph = (\nodes,\edges)$ with edge costs $\cost\colon \edges \rightarrow \naturals$, 
	compute the Minimum Spanning Tree of \graph.
\end{problem}

\begin{problem}[\stepthree] 
	\label{def:ps}
	Given a graph $\graph = (\nodes,\edges)$ with edge costs $\cost\colon \edges \rightarrow \naturals$, 
	a rooted forest, and a set of marked nodes $\thesubset\subseteq \nodes$, 
	select the forest edges connecting each marked node to its~root. 
\end{problem}
\Cref{def:ps} can be reduced to \emph{approximate transshipment} as follows:
\begin{compactenum}[(i)]
  \item count the number of marked nodes in each tree;
  \item set the demand of each marked node to~$-1$ and the demand of the root of each tree to the number of marked nodes in the tree;
  \item set edge costs to $0$ for tree edges and to~$+\infty$ for all other edges;
  \item solve the approximate transshipment problem for these demands and edge weights; and
  \item select all edges with non-zero flow in the output.
\end{compactenum}
Note that the only non-trivial step of the reduction is the computation of the demand, which boils down to a single Partwise Aggregation.
\begin{problem}[Edge-Cost Reduction]
	\label{def:ecr}
	Given a graph $\graph = (\nodes,\edges)$ 
	with edge costs $\cost\colon \edges \rightarrow \naturals$,
	a radius $\radius$,
	and an output of \Cref{def:aSSSP},
	compute, for each edge $\edge\in\edges$,
	\begin{align*}
		\reducedcost(\edge) = \max\left\{0,\cost(\edge) - \sum_{\node\in \edge}\max\{\radius-\distance_{\forest'}(\thesubset,\node),0\}\right\}\;.
	\end{align*}
\end{problem}

\begin{problem}[Forest-Function Evaluation]
	Given a graph $\graph = (\nodes,\edges)$, 
	a partition $\components$ of the node set such that each $\component \in \components$ is connected, 
	and a proper forest function $\myfunction\colon 2^\nodes\rightarrow\{0,1\}$, 
	evaluate $\myfunction(\component)$ for each component $\component \in \components$.
	\label{def:ffe}
\end{problem}

\clearpage

\subsubsection{Meta-Algorithm Using the Building Blocks}
\label{sec:algorithm:modelagnostic:buildingblocks}
Initialize $\forest$, $\components$, $\terminals^1$, $\reducedcost$, and $\radius$ as in \Cref{alg:gw-clean}.
Throughout our algorithm,
we maintain a set of connected components $\components$ with activity statuses $\myfunction(\component)$ for each $\component\in\components$.
At the beginning of phase~$0$, 
\components contains exactly the singleton sets corresponding to all nodes, i.e., $\components = \{\{\node\}\mid \node\in\nodes\}$,
and the active components are the terminals. 
Each phase~$i$ of our algorithm (i.e., one loop iteration in \Cref{fig:overview}, simulating one while-loop iteration of \cref{alg:gw-clean}) 
then consists of the following steps, 
executed for $\radius = (1+\overgrowth)^i\cdot\frac{\overgrowth}{4}$.
\begin{enumerate}[label=(\arabic*)]
	\item \emph{Approximate Set-Source Shortest-Path Forest (aSSSP).}\label{step:aSSSP}
	Assign as temporary edge weight to edge $\edge\in \edges$
	the reduced cost $\reducedcost(\edge)$ 
	if $\reducedcost(\edge)>0$ 
	or $\edge \in \forest\cup \forest'$ (where $\forest'\coloneq\emptyset$ in phase~$0$).
	Compute a $(1+\distanceeps)$-approximate (\radius-restricted) set-source shortest-path forest $\forest'$, 
	using the active terminals $\terminals^1 = \{\min\{\node\mid \node\in\terminals\cap\component\}\mid \component \in \components, \myfunction(\component) = 1\}$ as sources, i.e., for each active component, $\terminals^1$ contains the terminal with the minimum identifier. 
	After Step~\ref{step:aSSSP}, 
	for each node $\node\in\nodes\setminus\terminals^1$ we know its parent in the truncated SSSP forest, 
	its closest source $\othernode\in\terminals^1$ in the respective shortest-path tree (if any), 
	and its distance $\distance_{\forest'}(\othernode,\node)$ to that source. 
	\item \emph{Edge-Cost Reduction (ECR).}\label{step:ecr} 
	Using the approximate \radius-restricted SSSP forest and the distances computed in Step~\ref{step:aSSSP}, 
	update the edge costs in accordance with \Cref{def:ecr}.\footnote{%
		We can keep the edge costs in \naturals by making sure that phases end with integral values of \radius. 
		Note that $\distanceeps \geq \nnodes^{-\BO(1)}$, 
		or distance computations must be exact. 
		Scale all weights by $\lceil\nicefrac{1}{\distanceeps}\rceil$. 
		Now rounding $\radius$ up to the next integer has marginal impact on the approximation guarantee, 
		as overgrowing by factor $(1+\distanceeps)$ plus an additive $1$ is not worse than overgrowing by factor $(1+2\distanceeps)$.
	}
	\item \emph{Candidate-Merge Identification (CMI).}\label{step:cmi}
	Identify the candidate merges \mergeedges between adjacent trees of the aSSSP forest computed in Step~\ref{step:aSSSP}, using that parents of nodes and reduced edge costs are known.  
	\item \emph{Minimum Spanning Tree (MST).}\label{step:msf}
	Compute a Minimum Spanning Tree \tree of \graph with the following edge weights:
	\begin{inparaenum}[(i)]
	\item $0$ for edges in $\forest'$, i.e., the tree edges in the output of Step~\ref{step:aSSSP},
	\item $1$ for edges in $\mergeedges$, i.e., those determined in Step~\ref{step:cmi}, and
	\item $+\infty$ (or a large value) for all other edges.
	\end{inparaenum}
	Mark all selected edges of \tree that are also in the set \mergeedges known from Step~\ref{step:cmi}, i.e.,	
	the edges constituting \augmentation, and add them to \forest 
	(thus \emph{excluding} all edges with temporary weight greater than $1$). 
	For each connected component $\component'$ of the forest constituted by the selected edges of temporary weight~$0$ or~$1$ that contains a terminal $\tau\in\component'$, 
	set $\min\{\node\in\component'\mid \myfunction(\{\node\}) = 1\}$ as the new identifier of the component \component to be created from $\component'$, 
	making it known to all $\node\in \component'$.
	\item \emph{\stepthree (\stepthreeabbrv).}\label{step:rps}
	Connect the marked edges identified in Step~\ref{step:msf} to the roots (i.e., the node with the same identifier as the component) of the components they connect by adding the necessary edges to \forest.
	\item \emph{Forest-Function Evaluation (FFE).}\label{step:ffe}
	Using the new component memberships known from Step~\ref{step:msf},
	update the set~\components and evaluate $\myfunction(\component)$ for each updated component $\component\in\components$.
	If $\myfunction(\component)=0$ for all such components, terminate and output \forest.
	Otherwise, continue with the next loop iteration.
\end{enumerate}

\clearpage

\subsection{Correctness and Complexity of our Model-Agnostic Meta-Algorithm}
Because \Cref{alg:gw-clean} computes a $(2+\epsilon)$-approximation by \Cref{thm:metatheorem}, 
we can prove the correctness and approximation guarantee of our model-agnostic variant (\cref{spec:meta-algo}) by showing its equivalence to \Cref{alg:gw-clean}.

\begin{restatable}[Model-Agnostic Shell-Decomposition Algorithm]{theorem}{equivalenceproof}
	\label{thm:modular-gw}
	For $\epsilon,\distanceeps,\overgrowth$ as in \Cref{thm:gwcleanapx}, a graph $\graph = (\nodes,\edges)$, edge costs $\cost\colon\edges\to\naturalsnozero$, 
	and proper function \myfunction, 
	the modular shell-decomposition algorithm described in \cref{sec:algorithm:modelagnostic:buildingblocks} yields a feasible $(2+\epsilon)$-approximation to the optimal solution. 
\end{restatable}
\begin{proof}
We prove the claim by induction on the phase~$i$, going step by step through the algorithm given in \Cref{sec:algorithm:modelagnostic:buildingblocks} and arguing why the computed objects, in particular \forest, match those of \Cref{alg:gw-clean}.

We augment the induction hypothesis by the claim that at the beginning of a phase, in \Cref{alg:gw-clean}, \edges contains exactly the edges of non-zero reduced cost and $\forest_{i-1}\cup \forest'_{i-1}$.
The induction anchor (phase $i=-1$) is given by the identical initialization of objects.
For the step to phase~$i\in \naturals$, observe first that by the induction hypothesis (in particular the additional claim), Step~\ref{step:aSSSP} computes the same $\forest_i'$ and the same distances as \Cref{alg:gw-clean}, and hence Step~\ref{step:ecr} yields the same $\reducedcost_i$.
It follows that Step~\ref{step:cmi} computes the same set $\mergeedges$ of candidate merges, 
implying that Step~\ref{step:msf} correctly determines $\augmentation_i$ and adds it to $\forest$.
Note that the latter step also updates component memberships and component identifiers, 
but does not yet evaluate whether $\myfunction(\component)=1$ for the new components.
This is finally done in Step~\ref{step:ffe}, 
such that in Step~\ref{step:aSSSP} of the next phase, the correct set $\terminals^1_i$ will be used.
It remains to prove the additional claim that $\edges_i$ contains exactly the edges of non-zero reduced cost and $\forest_{i-1}\cup \forest'_{i-1}$, which now is immediate from Line~\ref{line:remove-unneeded} and \Cref{lem:gwcleanbasic}~(vii).

We conclude that both algorithms terminate at the end of the same phase $j$, returning the same forest $\forest=\forest_j$, which by \Cref{thm:gwcleanapx} is a $(2+\epsilon)$-approximation.
\end{proof}

The proof of our main theorem, then, follows immediately.

\metatheorem*
\begin{proof}
	By \cref{thm:modular-gw},
	our modular shell-decomposition algorithm (\cref{alg:gw-clean}) delivers the desired approximation guarantee.
	Without loss of generality, we may assume that $\epsilon\in \nnodes^{-\BO(1)}$, as this is enough to enforce that $\epsilon$ times the cost of an optimal solution is smaller than $1$, i.e., a $2$-approximation is guaranteed.
	Thus, it is sufficient to instantiate the algorithm with $\distanceeps,\overgrowth\in\nnodes^{-\BO(1)}$,
	such that by \cref{lem:gwcleantermination}, the algorithm will terminate after $\BO(\frac{\log\nnodes}{\overgrowth}) = \BO(\frac{\log\nnodes}{\epsilon}) = \tildeO(\epsilon^{-1})$ while-loop iterations. 
	In each of these iterations (up to bookkeeping operations), 
	aSSSP, MST, RPS, and FFE computations are performed exactly once, 
	yielding the desired model-agnostic complexity. 
\end{proof}

%% file: text/gw-original.tex
\begin{algorithm2e}[t]
	\small
	\DontPrintSemicolon
	\caption{\mbox{\GW algorithm for $(2-\nicefrac{2}{\nterminals})$-Approximation of Constrained Forest Problems \cite{goemans1995approximation}.}}
	\label{alg:gw-original}
	\KwIn{%
		An undirected graph $\graph = (\nodes,\edges)$, 
		edge costs $\cost: \edges\rightarrow \naturalsnozero$, 
		and a proper function \myfunction}
	\KwOut{A forest $\forest$ and a value \lowerbound}
	$\forest'\gets \emptyset$\; 
	\lowerbound$\gets 0$\tcp*{Implicitly set $\yvar_\thesubset\gets 0$ for all $\thesubset\subset\nodes$}
	\components$\gets \{\{\node\}\mid \node \in \nodes\}$\;
	\ForEach{$\node\in\nodes$}{$\radius(\node)\gets 0$}
	\While{$\exists\component\in\components:\myfunction(\component) = 1$}{%
		$\edges_\components \gets \{\{i,j\}\in\edges\mid
		i\in\component_i\in\components,
		j\in\component_j\in\components,
		\component_i\neq \component_j,
		\myfunction(\component_i) + \myfunction(\component_j)>0\}$\;
		$\growth,\edge \gets \underset{\edge=\{i,j\}\in\edges_\components}{\min,\argmin}\left\{
		\frac{\cost'(\edge)}{\myfunction(\component_i)+\myfunction(\component_j)} \mid \cost'(\edge) = \cost(\edge)-\radius(i)-\radius(j)
		\right\}$\;
		$\forest'\gets\forest'\cup\{\edge\}$\;
		\ForEach{$\node\in\component_r\in\components$}{%
			$\radius(\node)\gets \radius(\node) + \growth\cdot\myfunction(\component_r)$\;
		}
		$\lowerbound\gets\lowerbound+\growth\cdot \sum_{\component\in\components}\myfunction(\component)$
		\tcp*{Implicitly set $\yvar_\component \gets \yvar_\component + \growth \cdot \myfunction(\component)~\forall \component\in\components$}
		$\components\gets\left(\components\cup\{\component_i\cup\component_j\}\right)\setminus\{\component_i,\component_j\}$\;
	}
	$\forest\gets \{\edge\in\forest'\mid \myfunction(N) = 1~\text{for some connected component $N$ of}~(\nodes,\forest'\setminus\{\edge\})\}$\;\label{line:final-filter}
	\KwRet{$\forest, \lowerbound$}
\end{algorithm2e}

%% file: text/gw-graph-clean.tex
\begin{algorithm2e}[H]
	\small
	\DontPrintSemicolon
	\caption{Shell-Decomposition Algorithm for $(2+\epsilon)$-Approximation of Constrained Forest Problems.
	}\label{alg:gw-clean}
	\KwIn{%
		A graph $\graph = (\nodes, \edges)$, 
		edge costs $\cost: \edges\rightarrow \naturalsnozero$, 
		and a proper forest function \myfunction
		}
	\KwOut{A forest $\forest$ and a value \lowerbound}
	$\forest\gets \emptyset$\; 
	\components$\gets \{\{\node\}\mid \node \in \nodes\}$\;
	$\terminals^1 \gets \left\{\node\mid \myfunction(\{\node\}) = 1\right\}$\tcp*{At the beginning, $\cardinality{\terminals^1} = \nterminals$}\label{line:active-terminals-start}
	\ForEach{$\edge\in\edges$}{
		$\reducedcost(\edge)\gets \cost(\edge)$
		\tcp*{$\reducedcost$ keeps track of \emph{reduced} costs}
	}
	$\lowerbound\gets 0$\; 
	$\radius \gets \frac{\overgrowth}{4}$\tcp*{For upper bound (edge selection), we could start with any $\radius \le \nicefrac{1}{2}$}\label{line:init-radius}
	\While(){$\exists\component\in\components:\myfunction(\component) = 1$}{%
		$\ballunion \gets$~Union of $(1+\distanceeps)$-approximate balls of radius $\radius$ around nodes in $\terminals^1$ under edge costs \reducedcost\footnotemark\;\label{line:growth-balls}
		$\forest' \gets (1+\distanceeps)$-approximate SSSP forest for edge cost $\reducedcost$ with set source $\terminals^1$, truncated to \ballunion\;\label{line:growth-sptree}
		\ForEach{$\edge\in\edges$}{
			$\reducedcost(\edge) \gets \max\left\{0,\reducedcost(\edge) - \cost_\ballunion(\edge)\right\}\;$\label{line:reduced-costs}\tcp*{$\cost_\ballunion(\edge) = \sum_{\node\in \edge}\max\{\radius-\distance_{\forest'}(\terminals^1,\node),0\}$
				}
		}
		$\reducedcomponents\gets$ Connected components of $(\nodes,\forest')$ containing a (unique) $\terminal\in \terminals^1$\;
		$\mergeedges \gets \{\{\othernode,\node\}\in \edges\mid \othernode\in \component_{\othernode}\in\reducedcomponents,\node\in \component_{\node}\in\reducedcomponents,\component_{\othernode}\neq \component_{\node},\reducedcost(\{\othernode,\node\})=0\}$\tcp*{Merge candidates}\label{line:merge-candidates}
	    $\augmentation\gets$ Arbitrary subset of $\mergeedges$ such that $\forest'\cup \augmentation$ is a forest spanning $(V,\forest'\cup \mergeedges)$\;\label{line:acyclic-sequence}
	    $\forest \gets \forest \cup \augmentation \cup \bigcup_{\node\in \edge\in \augmentation}\{\somepath_\node\mid \somepath_\node$ is the (unique) path in $\forest'$ from \node to some $\terminal \in \terminals^1\}$ \tcp*{Merge}\label{line:perform-merges}
		$\edges\gets (\edges\setminus \{\edge\in \edges\mid \reducedcost(\edge)=0\})\cup \forest'\cup \augmentation$\tcp*{Remove unneeded edges contained in \ballunion}\label{line:remove-unneeded}
	    $\components \gets$ Connected components of $(\nodes,\forest)$ \tcp*{Update connected components}\label{line:components-update}
	    $\terminals^1 \gets \left\{\min\{\node\in\component\mid \myfunction(\{\node\}) = 1\} \mid \myfunction(\component) = 1, \component\in\components\right\}$\tcp*{Update active terminals}\label{line:active-terminals-update}
	    $\lowerbound\gets \lowerbound + (1+\distanceeps)^{-1}\radius\cardinality{\terminals^1}$\;\label{line:lobo-update}
	    $\radius \gets (1+\overgrowth)\radius$\tcp*{Update ball radius}\label{line:increase-radius}
	}
	\KwRet{$\forest, \lowerbound$}
\end{algorithm2e}
\footnotetext{%
	\ballunion as used in \cref{line:growth-balls} contains all parts of edges in a $(1+\epsilon)$-approximate SSSP forest that lie at distance at most~\radius from their respective roots. 
	Note that \ballunion can also contain \emph{parts of} edges, 
	whereas $\forest'$ (\cref{line:growth-sptree}), 
	a $(1+\epsilon)$-approximate SSSP forest truncated to~\ballunion, 
	only contains full edges. 
}

%% file: text/appendix-distributed.tex
\section{Distributed Algorithm}
\label{apx:distributed}
\label{apx:algorithm:distributed}

\subsection{Computational Model}
\label{congest:model}

In the $\Congest(\log \nnodes)$ model (\Congest model) of distributed computing \cite{peleg2000distributed}, 
the input graph \graph also models the topology of a distributed system. 
Each node initially knows its unique $\BO(\log\nnodes)$-bit identifier, 
the identifiers of its neighbors, the weight of its incident edges, and its local problem-specific input. 
Algorithms in the \Congest model proceed in synchronous rounds. 
In each round, each node may
\begin{inparaenum}[(1)]
	\item perform arbitrary,
	finite computations based on its local information, 
	\item send one message of $\BO(\log \nnodes)$ bits to each of its neighbors (where the messages sent to different neighbors may be distinct), and
	\item receive all messages sent by its neighbors.
\end{inparaenum}
The crucial complexity measure for \Congest algorithms is their \emph{round complexity}, 
i.e., the number of \emph{rounds} until all nodes have terminated explicitly.\footnote{%
	A secondary complexity measure for \Congest algorithms 
	is their \emph{message complexity}, i.e., the number of \emph{messages} sent across edges sent before termination. 
	For simplicity, we omit this measure in our exposition. 
}

As a generalization of the model, we will allow for (a constant number of) \emph{virtual nodes.}
Their incident edges are not part of the communication topology, and algorithms need to simulate the computations performed by these nodes.
We require that, initially, the weights of the edges incident with virtual nodes 
are known to their non-virtual endpoints, 
while other inputs to virtual nodes and weights of edges between virtual nodes should be known globally.
This addition enables us to apply our results to a wider range of tasks, e.g., \facloclong.

Note that all tasks studied in this work require $\Omega(\hopdiameter)$ rounds, even if $(V,E)$ is known to all nodes, i.e., $\Omega(\hopdiameter)$ is a universal lower bound.
Moreover, simple standard techniques can be used to chain together subroutines without knowing when they terminate at which node in advance.
For simplicity, we omit the respective book-keeping and control from the description of our algorithms.

\subsection{Meta-Algorithm}
\label{congest:meta}

Implementing our model-agnostic algorithm in the \Congest model, 
we obtain the following unconditional result. 

\begin{theorem}\label{thm:congest-cfp-existential}
	In the \Congest model with $\BO(1)$ virtual nodes, for any $0<\epsilon\le 1$, a $(2+\epsilon)$-approximation to any proper 
	Constrained Forest Problem can be obtained in $\tildeO(\epsilon^{-3}(\sqrt{\nnodes} + \hopdiameter)+\epsilon^{-1}\myfunctioncomplexity)$ rounds, 
	where \myfunctioncomplexity is the complexity of evaluating \myfunction in the distributed setting.
\end{theorem}

\begin{proof}
To prove this theorem, 
it suffices to implement the five problem-independent building blocks of our model-agnostic algorithm  (i.e., blocks (1)--(5), cf.~\cref{sec:algorithm:modelagnostic:buildingblocks}) 
with complexity $\tildeO(\epsilon^{-2}(\sqrt{\nnodes} + \hopdiameter))$ in the \Congest model;
the claim then follows from \Cref{thm:metatheorem}.
Outside of subroutine calls, all nodes simulate the virtual nodes, i.e., 
they maintain their state, 
with the exception of not necessarily knowing the weights of their edges to non-virtual nodes.
If not all nodes have the necessary information, this can be fixed in $\BO(\hopdiameter)$ rounds by broadcasting the respective information globally.

\begin{enumerate}
	\item \emph{Distributed $\alpha$-Approximate Set-Source Shortest-Path Forest.}
	\citet{rozhon2022paths} provide a deterministic \Congest algorithm computing a $(1+\epsilon)$-approximate set-source shortest-path forest in $\tildeO((\sqrt{\nnodes} + \hopdiameter)\epsilon^{-2})$ time for $\epsilon\in (0,1]$, 
	with which we can also compute an \radius-restricted approximate set-source shortest-path forest $\forest'$.
	This algorithm operates in the minor aggregation model and is capable of handling $\BO(1)$ virtual nodes.
	
	\item \emph{Distributed Edge-Cost Reduction.}
	This problem can be solved via local computation based on knowledge from Step~1, 
	i.e., each $\node\in\nodes$ can compute the reduced cost of its incident edges based on local information.
	
	\item \emph{Distributed Candidate-Merge Identification.}
	Each non-virtual node sends the identifier of its root to all neighbors (one round);
	for virtual nodes, the identifier is known to all nodes.
	Now it can be locally determined for each incident edge (and those with two virtual endpoints) whether it is a merge candidate.
	\item \emph{Distributed Minimum Spanning Tree.} 
	\citet{kutten98distributed} provide a deterministic MST algorithm running in $\tildeO(\sqrt{\nnodes} + \hopdiameter)$ rounds.
	It is not difficult to modify this algorithm to handle $\BO(1)$ virtual nodes.
	For an explicit solution based on Partwise Aggregation, see \Cref{apx:msf}.
	\item \emph{Distributed \stepthree.}
	\citet{rozhon2022paths} provide a deterministic \Congest algorithm for computing $(1+\epsilon)$-approximate Transshipment in $\tildeO((\sqrt{\nnodes} + \hopdiameter)\epsilon^{-2})$ time for $\epsilon\in (0,1]$, i.e., $\tildeO(\sqrt{\nnodes} + \hopdiameter)$ for $\epsilon\in\Omega(1)$.
	The algorithm can handle $\BO(1)$ virtual nodes.
	By the reduction pointed out in \Cref{sec:algorithm:modelagnostic:problems} and the fact that a single Partwise Aggregation can be carried out in $\BO(\sqrt{\nnodes}+\hopdiameter)$ rounds (see, e.g., \cite{rozhon2022paths}), the same time bounds extend to \Cref{def:ps}.\qedhere
\end{enumerate}
\end{proof}

Furthermore, the conditional results by \citet{rozhon2022paths} translate into the following (partly) conditional results. 

\begin{theorem}\label{thm:congest-cfp-universal}
	In the \Congest model with $\BO(1)$ virtual nodes, for any $0<\epsilon\le 1$,
	Constrained Forest Problems can be approximated up to a factor of $(2+\epsilon)$ in $\tildeO(\epsilon^{-3}\pacomplexity\nnodes^{o(1)} + \epsilon^{-1}\myfunctioncomplexity)$ rounds, 
	where \pacomplexity is the complexity of solving Partwise Aggregation 
	and \myfunctioncomplexity is the complexity of evaluating \myfunction in the distributed setting. 
\end{theorem}
\begin{proof}
	In \cref{thm:congest-cfp-existential}, the only steps requiring more than $\tildeO(D)\subseteq \tildeO(\pacomplexity)$ rounds were $(1+\epsilon)$-approximate Set-Source Shortest-Path forest, 
	Minimum Spanning Tree, and $(1+\epsilon)$-approximate Transshipment (for Root-Path Selection). 
	\citet{rozhon2022paths} provide deterministic algorithms for computing $(1+\epsilon)$-approximate Set-Source Shortest-Path forests 
	and $(1+\epsilon)$-approximate Transshipment in $\tildeO(\epsilon^{-2}\pacomplexity\nnodes^{o(1)})$ time. 
	Furthermore, \citet{ghaffari2016algorithms} give a randomized algorithm for computing a Minimum Spanning Forest (and, consequently, an MST) in $\tildeO(\pacomplexity)$ time, 
	and we can remove the randomness from this algorithm and handle virtual nodes as detailed in \cref{apx:msf}.
\end{proof}

\begin{theorem}\label{thm:congest-cfp-universal2}
	In the \Congest model with $\BO(1)$ virtual nodes, 
	if there exists a deterministic $(\tildeO(1),\tildeO(1))$-cycle-cover algorithm for $\tildeO(1)$-diameter graphs that runs in $\tildeO(1)$ rounds, for any $0<\epsilon\le 1$,
	Constrained Forest Problems can be approximated in $\tildeO(\epsilon^{-2}\pacomplexity + \myfunctioncomplexity)$ rounds. 
\end{theorem}
\begin{proof}
	If a deterministic $(\tildeO(1),\tildeO(1))$-cycle-cover algorithm for $\tildeO(1)$-diameter graphs that runs in $\tildeO(1)$ rounds exists, 
	the algorithms for computing $(1+\epsilon)$-approximate Set-Source Shortest-Path forests 
	and $(1+\epsilon)$-approximate Transshipment 
	provided by \citet{rozhon2022paths} 
	run in $\tildeO(\epsilon^{-2}\pacomplexity)$ time. 
	This suffices to shave off the $\nnodes^{o(1)}$ overhead incurred in \cref{thm:congest-cfp-universal}.
\end{proof}

Note that all algorithms used as subroutines in  \cref{thm:congest-cfp-existential,thm:congest-cfp-universal,thm:congest-cfp-universal2} are deterministic.

\subsection{Applications}

The solutions presented in the previous section work for \emph{any} proper Constrained Forest Problem in the \Congest model, 
and hence, 
their running time depends on the complexity of evaluating the forest function \myfunction. 
We now instantiate \myfunction with our concrete Constrained Forest Problems. 

\subsubsection{Steiner Forest}

There are several ways to specify the input of a Steiner Forest instance in the distributed setting that, 
perhaps surprisingly, 
lead to different complexities.

\paragraph{Steiner Forest--Input Components (SF--IC).}
Our first SF formulation, Steiner Forest--IC, assumes that each node knows the identifier of the input component to which it belongs. 
This formulation is closest to the model-agnostic formulation.

\begin{problem}[Steiner Forest--IC] 
	\label{def:sf-problem-ic}	 
	Given a graph $\graph = (\nodes,\edges)$, 
	edge costs $\cost\colon\edges\rightarrow \naturalsnozero$,
	and component identifiers ${\componentidentifier\colon \nodes\rightarrow [\ncomponents]\cup\{\bot\}}$, 
	where each node \node knows only its component identifier $\componentidentifier(\node)$, 
	compute a minimum-cost edge subset $\forest\subseteq\edges$ such that for each $i\in[\ncomponents]$, 
	$\inputcomponent_i = \{\node\in\nodes\mid\componentidentifier(\node) = i\}$ is connected in~$(\nodes,\forest)$.
\end{problem}

With this input specification, 
we incur an overhead of \ncomponents, 
which is existentially optimal~\cite{lenzen2014steiner}.  

\begin{lemma}[Distributed Complexity of Steiner Forest--IC]
	\label{result:sf-congest-ic}
	In the \Congest model, for any\linebreak $0<\epsilon\le 1$, Steiner Forest--IC can be $(2+\epsilon)$-approximated in $\tildeO(\epsilon^{-3}\min\{\pacomplexity\nnodes^{o(1)}, \sqrt{\nnodes}+\hopdiameter\}+\epsilon^{-1}\ncomponents)$ rounds deterministically, 
	where \ncomponents is the number of input components.
	This holds also in the presence of $\BO(1)$ virtual nodes.
\end{lemma}

\begin{proof}
	We propose an $\tildeO(\pacomplexity + \ncomponents)$-round algorithm performing Forest-Function Evaluation for SF--IC in the \Congest model. 
	Our solution consists of two main steps.
	\begin{enumerate}[label=(\arabic*)]
		\item Construct a global breadth-first-search (BFS) tree rooted at an arbitrary node $\rootnode \in \nodes$. 
		We identify active components through one convergecast phase toward the root, 
		followed by one broadcast phase from the root. 
		In the convergecast phase, for each input component $i\in [k]$,
		each node $\node \in \nodes$ sends
		\begin{compactenum}
		\item $(i,\bot)$ if there is no node $\othernode$ with $\componentidentifier(\othernode)=i$ in its subtree,
		\item $(i,\component)$ if each $\othernode$ with $\componentidentifier(\othernode)=i$ in its subtree lies in the same component $\component$ and there is at least one such node, and
		\item $(i,\times)$ if there are two nodes $\othernode$, $\othernode'$ in its subtree with $\componentidentifier(\othernode)=\componentidentifier(\othernode')=i$, but $\component_{\othernode}\neq \component_{\othernode'}$,
		\end{compactenum}
		up the BFS tree.
		Note that each node can determine its message to its parent from those of its children and its local input.
		In the subsequent broadcast phase, 
		\rootnode sends messages down the tree, 
		informing all nodes about which labels are contained in at least two different components. 
		Since there are \ncomponents different labels and the BFS tree has depth at most \hopdiameter, 
		the whole process can be completed in $\BO(\hopdiameter + \ncomponents)$ time using standard pipelining techniques. 
		
		\item Leveraging the information gathered in the previous step, 
		locally mark all nodes whose label is contained in at least two different components. 
		Now perform Partwise Aggregation for each $\component\in\components$ to identify the components containing at least one marked node, 
		and report these components as active. 
		This can be done in $\tildeO(\pacomplexity)$ time. 
	\end{enumerate}
	Recall that the state of virtual nodes is globally known, so their contribution to the output of the routine can be locally included into the computation by each node.
	The result of the computation for the virtual nodes is broadcast over the BFS tree in $\BO(\hopdiameter)$ rounds.
	
	Plugging this algorithm into \cref{thm:congest-cfp-existential,thm:congest-cfp-universal}, 
	the lemma follows. 
\end{proof}

\paragraph{Steiner Forest--Symmetric Connection Requests (SF--CR).}
Our second SF formulation, Steiner Forest--CR, 
assumes that each node knows the identifiers of the nodes to which it wants to connect. 
The input components, which are no longer given explicitly, 
are then the connected components of the graph implicitly defined by the connection requests.

\begin{problem}[Steiner Forest--CR] 
	\label{def:sf-problem-cr}	 
	Given a graph $\graph = (\nodes,\edges)$, 
	edge costs $\cost\colon\edges\rightarrow \naturalsnozero$, 
	and a set of connection requests $\requests_\node\subseteq\nodes\setminus\{\node\}$ at each node $\node\in\nodes$, 
	compute a minimum-cost edge subset $\forest\subseteq\edges$ such that for each $\node\in\nodes$ and $\othernode\in\requests_\node$, 
	\othernode and \node are connected in~$(\nodes,\forest)$.
\end{problem}

With this input specification, 
we incur an overhead of \nterminals,
which follows from a reduction in~\cite{lenzen2014steiner} (Lemma~2.3);
this is also shown to be existentially optimal~\cite{lenzen2014steiner}.  
To be self-contained and clarify that virtual nodes are not problematic, we provide a direct implementation here.
\begin{lemma}[Distributed Complexity of Steiner Forest--CR]
	\label{result:sf-congest-cr}
	In the \Congest model, for any ${0<\epsilon\le 1}$, Steiner Forest--CR can be $(2+\epsilon)$-approximated deterministically in $\tildeO(\epsilon^{-3}\min\{\pacomplexity\nnodes^{o(1)},$
	$\sqrt{\nnodes}+\hopdiameter\}+\epsilon^{-1}\nterminals)$ time, 
	where \nterminals is the number of terminals.
	This holds also in the presence of $\BO(1)$ virtual nodes.
\end{lemma}
\begin{proof}
	We propose an $\tildeO(\pacomplexity + \nterminals)$-round algorithm performing Forest-Function Evaluation for SF--CR in the \Congest model. 
	The procedure is similar to that described in \cref{result:sf-congest-ic}.
	\begin{enumerate}[label=(\arabic*)]
		\item Construct a global BFS tree rooted at an arbitrary node $\rootnode \in \nodes$. 
		Each terminal \terminal then sends a message $(\terminal,\component_\terminal)$ up the BFS tree, 
		and \rootnode disseminates this information to all nodes.
		This can be done in $\BO(\hopdiameter + \nterminals)$ time using pipelining.
		\item Using the information gathered in the previous step, check locally for each terminal \terminal if its connection requests $\requests_\terminal$ are satisfied (i.e., $\component_\terminal = \component_\node$ for all $\node\in\requests_\terminal$), 
		and locally mark all nodes with unsatisfied connection requests. 
		Now perform Partwise Aggregation for each $\component\in\components$ to identify the components containing at least one marked node, and report these components as active. 
		This can be done in $\tildeO(\pacomplexity)$ time. 
	\end{enumerate}
	Again, the state of virtual nodes is globally known and can be locally considered in determining the output of the routine by each node, and their output can be made globally known within $\BO(\hopdiameter)$ rounds.
	Plugging this algorithm into \cref{thm:congest-cfp-existential,thm:congest-cfp-universal}, 
	the lemma follows. 
\end{proof}
The $\epsilon^{-1}$ factor multiplied with $\nterminals$ can be removed by first applying the reduction to SF--IC given by \citet{lenzen2014steiner}, adapting the procedure to account for virtual nodes as above.

\paragraph{Steiner Forest--Symmetric Connection Requests (SF--SCR).}
Our third SF formulation, Steiner Forest--SCR, 
modifies Steiner Forest--CR to assume that the given connection requests are symmetric (i.e., if \othernode wants to connect to \node, \node also wants to connect to \othernode). 

\begin{problem}[Steiner Forest--SCR] 
	\label{def:sf-problem-scr}	 
	Given a graph $\graph = (\nodes,\edges)$, 
	edge costs $\cost\colon\edges\rightarrow \naturalsnozero$, 
	and a set of symmetric connection requests $\requests\subseteq\binom{\nodes}{2}$, 
	where each node \node knows $\requests_\node = \{\othernode\in\nodes\mid \{\othernode,\node\}\in\requests\}$, 
	compute a minimum-cost edge subset $\forest\subseteq\edges$ such that for each $\{\othernode,\node\}\in\requests$, 
	\othernode and \node are connected in~$(\nodes,\forest)$.
\end{problem}

While this specification is subject to a deterministic existential lower bound of $\tildeOmega(\nterminals)$, 
which we prove in \cref{apx:lower}, 
randomization allows us to circumvent this lower bound, 
shedding the overhead of \nterminals incurred by Steiner Forest--CR. 

\begin{lemma}[Distributed Complexity of Steiner Forest--SCR]
	\label{result:sf-congest-src}
	In the \Congest model, for any $0<\epsilon\le 1$, Steiner Forest--SCR can be $(2+\epsilon)$-approximated in $\tildeO(\epsilon^{-3}\min\{\pacomplexity\nnodes^{o(1)},\sqrt{\nnodes}+\hopdiameter\})$ by a randomized algorithm that succeeds with high probability.
	This holds also in the presence of $\BO(1)$ virtual nodes.
\end{lemma}

\begin{proof}
	We describe an $\tildeO(\pacomplexity)$-round randomized algorithm performing Forest-Function Evaluation for SF--SCR in the \Congest model.
	To check if a connected component \component remains active, we proceed as follows.
	We first describe the test assuming shared randomness and then use standard techniques to remove this restriction. 
	For each connection request $\{\othernode,\node\}$, 
	flip a fair independent coin and denote the result by $\cost_{\othernode,\node}$. 
	\component now computes ${\costsum_\component = \sum_{\othernode\in\component}\sum_{\node\in\nodes} \cost_{\othernode,\node}\bmod 2}$ via Partwise Aggregation. 
	If the result is $0\bmod 2$, the test is considered passed. 
	Note that if the component is inactive, the outcome is always $0\bmod 2$, regardless of the coin flips. 
	If it is active, let $\{\othernode,\node\}$ be such that $\othernode\in\component$ and $\node\in\nodes\setminus\component$. 
	Thus, $\cost_{\othernode,\node}$ contributes exactly once to $\costsum_\component$. 
	Since $\cost_{\othernode,\node}$ is an independent fair coin flip, 
	regardless of the other summands, 
	the test thus fails with probability $\nicefrac{1}{2}$.
	By performing the above test $\BO(\log\nnodes) = \tildeO(1)$ times, 
	with high probability, 
	we correctly determine the activity status of a component \component by checking whether it fails one of the tests. 
	Since we can do the $\BO(\log\nnodes)$ individual tests concurrently with $\BO(\log\nnodes)$-bit messages, 
	the overall operation can be done with $\BO(1)$ Partwise Aggregations.
	
	Finally, to replace shared randomness, we apply the probabilistic method to prove that sampling a small random seed and broadcasting it is sufficient.\footnote{%
		The approach we describe here is computationally inefficient, (ab)using that \Congest does not impose limits on local computation. 
		We chose it to enable a simple, self-contained presentation. 
		A computationally efficient alternative is to use polynomial hashing with a sufficiently large prime number; 
		see, e.g., \cite[Sec.~15.1]{harvey2023first}.
		}
	To this end, we describe an instance of the test by the node identifiers involved and the graph of connection requests issued between them.
	There are $2^{\BO(\nnodes\log \nnodes)}\cdot \frac{\nnodes(\nnodes-1)}{2}=2^{\BO(\nnodes\log \nnodes)}$ such instances.
	Now sample uniformly at random $N$ strings, where each string contains, for each pair of possible identifiers, sufficiently many random bits to execute the above test.
	Note that, for any specific instance, given such a string, 
	the test fails with probability $\varepsilon\in \nnodes^{-c}$ for some constant $c$ of our choosing.
	Fix such an instance and observe that, by construction, the expected number of sampled strings for which the test fails is $\varepsilon N$.
	By Chernoff's bound, the probability that it fails on $2\varepsilon N$ sampled strings is $2^{-\Omega(\varepsilon N)}$.
	Selecting $N\in \BO(\varepsilon^{-1}n\log n)$ and applying a union bound, we reach a non-zero probability that the test fails on \emph{none} of the possible instances with $n$ nodes for more than a $2\varepsilon$-fraction of the sampled strings.
	
	Using this observation, we can proceed as follows.
	We determine $n$ and make it known to all nodes (within $\BO(\hopdiameter)$ rounds).
	Then each node locally and deterministically computes the same list of $N\in \BO(\varepsilon^{-1}\nnodes\log \nnodes)$ strings satisfying that for no possible instance on $\nnodes$ nodes, the test fails when using them as random bit on more than a $2\varepsilon$-fraction of this list.
	Finally, a leader (say the node with the smallest identifier) samples uniformly from this list and broadcasts its choice, encoding it with $\BO(\log N)=\BO(\log (\nnodes^{\BO(1)}))=\BO(\log \nnodes)$ bits;
	flooding this random seed through the network takes $\BO(\hopdiameter)$ additional rounds, not affecting the asymptotic complexity.
	
	This algorithm extends to up to $\BO(1)$ virtual nodes as before: 
	The contribution of these nodes is locally included in the computation, and their output is broadcast to all nodes within $\BO(\hopdiameter)$ rounds.
	The claim now follows from \cref{thm:congest-cfp-existential,thm:congest-cfp-universal}. 
\end{proof}

\paragraph{Steiner Forest--Cardinality Input Components (SF--CIC)}
Our fourth SF formulation, Steiner Forest--CIC, 
is essentially Steiner Forest--IC, 
but with additional information on the cardinality of each input component. 

\begin{problem}[Steiner Forest--CIC] 
	\label{def:sf-problem-cic}	 
	Given a graph $\graph = (\nodes,\edges)$, 
	edge costs $\cost\colon\edges\rightarrow \naturalsnozero$,
	and component identifiers ${\componentidentifier\colon \nodes\rightarrow [\ncomponents]\cup\{\bot\}}$, 
	where each node \node knows its component identifier $\componentidentifier(\node)$ and the cardinality of its input component, 
	compute a minimum-cost edge subset $\forest\subseteq\edges$ such that for each $i\in[\ncomponents]$, 
	$\inputcomponent_i = \{\node\in\nodes\mid\componentidentifier(\node) = i\}$ is connected in~$(\nodes,\forest)$.
\end{problem}

While this specification is subject to a deterministic lower bound of $\tildeOmega(\ncomponents)$, which we prove in \cref{apx:lower}, 
randomization allows us to circumvent this lower bound, 
removing the overhead of \ncomponents from the running time incurred by Steiner Forest--IC. 

\begin{lemma}[Distributed Complexity of Steiner Forest--CIC]
	\label{result:sf-congest-cic}
	In the \Congest model, for any $0<\epsilon\le 1$, Steiner Forest--CIC can be $(2+\epsilon)$-approximated in $\tildeO(\epsilon^{-3}(\sqrt{\nnodes}+\hopdiameter)+\epsilon^{-1}\nnodes^{\nicefrac{2}{3}})$ time using randomization.
	This holds also in the presence of $\BO(1)$ virtual nodes. 
\end{lemma}
\begin{proof}
	We describe an $\tildeO(\nnodes^{\nicefrac{2}{3}}+\hopdiameter)$-round randomized algorithm performing Forest-Function Evaluation for SF--CIC in the \Congest model. 
	Leveraging the nodes' knowledge of the cardinality of their input components, 
	as well as the fact that we can determine the cardinality of a current component \component via a simple Partwise Aggregation, 
	the algorithm essentially uses the strategy from SF--IC to handle input components of size at least $\nnodes^{\nicefrac{1}{3}}$, 
	and a strategy similar to that from SF--SCR to handle smaller input components. 
	First, all components determine their size using partwise aggregation, where $\pacomplexity\in \BO(\sqrt{\nnodes}+\hopdiameter)$.
	Then each component $\component$ of size at most $\nnodes^{\nicefrac{2}{3}}$ computes, 
	for each input component, the size of its intersection with $\component$, 
	using $\BO(\nnodes^{\nicefrac{2}{3}})$ rounds and pipelining internally.
	On the other hand, for each input-component size $s < \nnodes^{\nicefrac{1}{3}}$, 
	each component of size larger than $\nnodes^{\nicefrac{2}{3}}$ tests all input components of size $s$ concurrently using a randomized strategy, 
	aggregating random coin flips in a manner similar to our approach in \cref{result:sf-congest-src}. 
	To this end, again assume shared randomness and flip a fair independent coin $\cost_{\componentidentifier}$ for each input-component identifier $\componentidentifier$. 
	Now, each component aggregates ${\costsum_\component = \sum_{\node \in \terminals \cap \component}\cost_{\componentidentifier(\node)}}\bmod s$, 
	where correctness is shown analogously to SF--SCR. 
	This can be done in $\BO(\nnodes^{\nicefrac{2}{3}}+\hopdiameter)$ rounds, 
	as we can pipeline over a global BFS tree with this congestion 
	(at most $\nnodes^{\nicefrac{1}{3}}$ input components, $\nnodes^{\nicefrac{1}{3}}$ component sizes, and $\BO(\log\nnodes)$ one-bit sums, which can be sent concurrently).
	The shared randomness can be replaced by broadcasting a random seed of size $\BO(\log \nnodes)$, as before.
	Finally, for each input component of size $s \geq \nnodes^{\nicefrac{1}{3}}$, 
	we check if there exist at least two different components $\component_i\neq\component_j$, 
	each containing at least one of its terminals, 
	by pipelining over the global BFS tree as we did for SF--IC. 
	This takes an additional $\BO(\nnodes^{\nicefrac{2}{3}}+\hopdiameter)$ rounds.
	
	This algorithm extends to up to $\BO(1)$ virtual nodes as before: 
	The contribution of these nodes is locally included in the computation, and their output is broadcast to all nodes within $\BO(\hopdiameter)$ rounds.
	Applying \cref{thm:congest-cfp-existential}, the lemma follows. 
\end{proof}

Note that the upper bound derived in \cref{result:sf-congest-cic} is not tight, 
and we could improve it by selecting cutoffs more judiciously and performing tests more efficiently, at the cost of adding complexity to the presentation. 
Since the main point here is to highlight the gap between the deterministic and randomized complexities for the global Steiner Forest problem, 
we leave the details to an extended version of this paper. 

\subsubsection{Point-to-Point Connection}
In the \Congest model, we can evaluate the forest function for the Point-to-Point Connection problem in $\tildeO(\pacomplexity)$ time 
by, for each component, \begin{inparaenum}[(1)]
	\item computing the number of sources and targets it contains via Partwise Aggregation, and 
	\item locally checking whether the number of sources equals the number of targets. 
\end{inparaenum}

\begin{corollary}[Distributed Complexity of Point-to-Point Connection]
	\label{result:ppc-congest}
	In the \Congest model, for any $0<\epsilon\le 1$, Point-to-Point Connection can be $(2+\epsilon)$-approximated in $\tildeO(\epsilon^{-3}\min\{\pacomplexity\nnodes^{o(1)}, \sqrt{\nnodes}+\hopdiameter\})$ time, deterministically. 
\end{corollary}

\subsubsection{\facloclong}

As \facloclong has the forest function of SF with terminals $C\cup\{s\}$, where $s$ is the (sole) virtual node in the input graph and the input format is SF--IC, we can evaluate the forest function in $\tildeO(\pacomplexity)$, as shown in the proof of \Cref{result:sf-congest-ic}.

\begin{corollary}[Distributed Complexity of \facloc]
	\label{result:fl-congest}
	In the \Congest model, for any $0<\epsilon\le 1$, \facloc can be $(2+\epsilon)$-approximated in $\tildeO(\epsilon^{-3}\min\{\pacomplexity\nnodes^{o(1)}, \sqrt{\nnodes}+\hopdiameter\})$~time, deterministically.
\end{corollary}

%% file: text/appendix-parallel.tex
\section{Parallel Algorithm}
\label{apx:parallel}
\label{apx:algorithm:parallel}

\subsection{Computational Model}
In the Parallel Random-Access Machine model (PRAM model), 
multiple processors share one random-access memory to jointly solve a computational problem \cite{fortune1978parallelism}.
Since all contention models for concurrent access to the
same memory cell by multiple processors are equivalent up to small ((sub-)logarithmic) factors in complexity \cite{harris1994survey}, 
we assume that there is no contention. 
The computation can thus be viewed as a Directed Acyclic Graph (DAG), 
whose nodes represent elementary computational steps and whose edges represent dependencies. 
The sources of the DAG, then, represent the input. 
In the PRAM model, 
the crucial complexity measures are the total size of the DAG, called \emph{work,}
i.e., the sequential complexity of the computation, 
and the maximum length of a path in the DAG, 
called \emph{depth,}
i.e., the time to complete the computation with an unbounded number of processors executing steps at unit speed.

\subsection{Meta-Algorithm}

Implementing our model-agnostic algorithm in the PRAM model, 
we obtain the following result. 

\begin{theorem}
	In the PRAM model, for any $0<\epsilon\le 1$, 
	Constrained Forest Problems can be approximated up to a factor of $(2+\epsilon)$ in ${\tildeO(\epsilon^{-3}\nedges + \epsilon^{-1}\myfunctioncomplexity_w)}$ work and $\tildeO(\epsilon^{-3}+\epsilon^{-1}\myfunctioncomplexity_d)$ depth, 
	where $\myfunctioncomplexity_w$ and $\myfunctioncomplexity_d$ are the work and depth required to evaluate \myfunction in the parallel setting, respectively.
\end{theorem}

\begin{proof}
To prove this theorem, 
it suffices to implement the five problem-independent building blocks of our model-agnostic algorithm (i.e., blocks (1)--(5), cf.~\cref{sec:algorithm:modelagnostic:buildingblocks}) 
with factor $\epsilon^{-1}$ smaller work and depth in the PRAM model;
we then apply \Cref{thm:metatheorem}.

\begin{enumerate}
	\item \emph{Parallel $\alpha$-approximate Set-Source Shortest-Path Forest.}
	\citet{rozhon2022paths} provide a deterministic parallel algorithm computing a $(1+\epsilon)$-approximate shortest-path tree on graphs with non-negative edge weights 
	in $\tildeO(1)$ depth and $\tildeO(\epsilon^{-2}\nedges)$ work
	for $\epsilon \in (0,1]$. 
	We can use this algorithm to solve our \emph{set-source} shortest-path problem 
	by introducing a new node and connecting it to all terminals via edges of weight $0$, 
	thus turning the problem into a single-source shortest-path problem. 
	Thus, we can also compute an \radius-restricted approximate Set-Source Shortest-Path forest $\forest'$.
	
	\item \emph{Parallel Edge-Cost Reduction.}
	Using the knowledge from Step~1, we can compute the reduced costs of all edges according in $\BO(\nedges)$ work and $\BO(1)$ depth. 
	
	\item \emph{Parallel Candidate-Merge Identification.}
	Using the knowledge from Steps~1 and~2, we can mark candidate edges in $\BO(\nedges)$ work and $\BO(1)$ depth.
	
	\item \emph{Parallel Minimum Spanning Tree.}
	\citet{chong2001concurrent} provide a deterministic parallel algorithm computing an MST in $\BO((\nnodes + \nedges)\log\nnodes) = \tildeO(\nnodes + \nedges)$ work and $\BO(\log\nnodes)=\tildeO(1)$ depth.\footnote{%
		This is optimal up to a logarithmic factor in the work. 
		To the best of our knowledge, all currently known work- \emph{and} depth-optimal algorithms are randomized \cite{halperin2001optimal,pettie2002randomized}.
	}
	
	\item \emph{Parallel \stepthree.}
	\citet{rozhon2022paths} provide a deterministic PRAM algorithm computing a $(1+\epsilon)$-approximate transshipment in $\tildeO(1)$ depth and $\tildeO( \epsilon^{-2}\nedges)$ work
	for any $\epsilon \in (0,1]$, i.e., $\tildeO(m)$ work for $\epsilon \in \Omega(1)$.\qedhere
\end{enumerate}
\end{proof}

\subsection{Applications}

The algorithm presented in the previous section works for \emph{any} proper Constrained Forest Problem in the PRAM model, 
and hence, its running time depends on the complexity of evaluating the forest function \myfunction in the parallel setting. 
As we analyze below, instantiating \myfunction with our concrete example problems does not incur any overhead over the problem-agnostic steps of our parallel algorithm. 

\subsubsection{Steiner Forest}
We first observe that the other three input representations can be efficiently reduced to SF--IC.
This is trivial for SF--CIC.
For SF--CR, we perform an MST calculation on the graph $(\terminals,\edges_{\terminals})$, where $\edges_{\terminals}$ is the union of an arbitrary spanning tree with the connection requests, and edge weights are $1$ for edges coming from connection requests and $2$ for those that do not.
We then remove all edges of weight~$2$ from the output and obtain a spanning forest of the input components, from which it is straightforward to generate the SF--IC representation using $\BO(\nterminals)$ work and $\tildeO(1)$ depth.

The above graph can be constructed using $\BO(\terminals+\edges_{\terminals})$ work and $\BO(1)$ depth,
and the MST computation can be performed in $\tildeO(\terminals+\edges_{\terminals})$ work and $\tildeO(1)$ depth using the algorithm by~\citet{chong2001concurrent}.
Note that, up to logarithmic factors, this matches the input size;
it is also straightforward to show that the majority of connection requests need to be read to solve the problem with constant probability (and good approximation ratio).\footnote{Unless we have already determined that all terminals are in the same input component, distinguishing whether there are one or more input components requires to keep reading connection requests. If the instance does not allow to cheaply connect these terminal sets that may or may not be in the same input component, figuring this out is essential for obtaining a good approximation ratio.}

Finally, SF--SCR trivially reduceds to SF--CR.
With this in mind, by slight abuse of notation, in the following, we will refer to SF as a whole and tacitly assume that the input format is SF--IC.

We propose an $\tildeO(\nnodes)$ work and $\tildeO(1)$ depth parallel algorithm for performing Forest-Function Evaluation for Steiner Forest in the PRAM model. 
For each $i \in [\ncomponents]$, 
denote by $\nterminals_i = \cardinality{\{\node\in \nodes\mid \componentidentifier(\node) = i\}}$ the number of nodes $\node\in\nodes$ with $\componentidentifier(\node) = i$. 
We can compute the $\nterminals_i$ for each $i\in[\ncomponents]$ by keeping track of $(i, count_i)$-pairs and performing parallel merge sort, 
which requires $\tildeO(n)$ work (as there are at most $\tildeO(\nnodes)$ distinct labels) and $\tildeO(1)$ depth \cite{cole1988parallel}.\footnote{%
	As $\nterminals_i$ remains constant for each $i\in[\ncomponents]$ over the course of the algorithm, 
	this needs to be done only once overall. 
}
Using the same method and the fact that we know the current component identifier of each node from Step~3 (MST), 
we can also count how often each label $i\in[\ncomponents]$ occurs in a component $\component\in\components$. 
Note that there exist at most \nnodes nonzero $(\component, i, count_i)$-triples. 
Therefore, 
using one step of computation by \nnodes parallel processors,  
we can evaluate, for each $\component\in\components$,
whether there exists a label $i\in[\ncomponents]$ such that $\nterminals_i \neq \cardinality{\{\node\in\component\mid \componentidentifier(\node) = i\}} > 0$, 
and record the result as the new activity status of each component. 

\begin{corollary}[Parallel Complexity of Steiner Forest]
	\label{result:sf-pram}
	In the PRAM model, for any $0<\epsilon\le 1$, Steiner Forest can be $(2+\epsilon)$-approximated with $\tildeO(\epsilon^{-3}\nedges)$ work and $\tildeO(\epsilon^{-3})$ depth.
\end{corollary}

\subsubsection{Point-to-Point Connection}
We propose an $\tildeO(\nnodes)$ work and $\tildeO(1)$ depth algorithm for performing Forest-Function Evaluation for Point-to-Point Connection in the PRAM model. 
For each node $\node\in\nodes$, 
set $x_\node = 1$ if \node is a source node, 
$x_\node = -1$ if \node is a destination node, 
and $x_\node = 0$ otherwise. 
In one step of computation by $\BO(\nnodes)$ parallel processors, 
calculate $\sum_{\node\in\component} x_\node$ for each component $\component\in\components$, 
and report all components with a nonzero sum as active components.

\begin{corollary}[Parallel Complexity of Point-to-Point Connection]
	\label{result:ppc-pram}
	In the PRAM model, for any $0<\epsilon\le 1$, Point-to-Point Connection can be $(2+\epsilon)$-approximated with $\tildeO(\epsilon^{-3}\nedges)$ work and $\tildeO(\epsilon^{-3})$ depth.
\end{corollary}

\subsubsection{\facloclong}

We add the virtual node $s$ and its edges to the input with $\BO(\nnodes)$ work and $\BO(1)$ depth.
The result now follows from \Cref{result:sf-pram}.
\begin{corollary}[Parallel Complexity of \facloc]
	\label{result:fl-pram}
	In the PRAM model, for any $0<\epsilon\le 1$, \facloc can be $(2+\epsilon)$-approximated in $\tildeO(\epsilon^{-3}\nedges)$ work and $\tildeO(\epsilon^{-3})$ depth.
\end{corollary}

%% file: text/appendix-streaming.tex
\section{Multi-Pass Streaming Algorithm}
\label{apx:streaming}

\subsection{Computational Model}
In streaming models of computation, 
an input graph is presented as a stream of edges, 
which are commonly assumed to arrive in \emph{arbitrary} order. 
Algorithms are assessed by the number of \emph{passes} they need to make over the stream and the amount of \emph{space} they require 
(organized in memory words of $\BO(\log \nnodes)$ bits). 
While streaming models differ with regard to the restrictions they impose on either measure, 
in this work, 
we consider the \emph{Multi-Pass (semi-)Streaming model} \cite{feigenbaum2005graph}, 
where we are allowed several passes over the stream and $\tildeO(\nnodes)$ space.

\subsection{Meta-Algorithm}

Implementing our model-agnostic algorithm in the Multi-Pass Streaming model, 
we obtain the following result. 

\begin{theorem}
	In the Multi-Pass Streaming model, for any $0<\epsilon\le 1$,
	Constrained Forest Problems can be approximated up to a factor of $(2+\epsilon)$ using $\tildeO(\nnodes + \myfunctioncomplexity_s)$ space and $\tildeO(\epsilon^{-3}+\epsilon^{-1}\myfunctioncomplexity_p)$ passes, 
	where $\myfunctioncomplexity_s$ and $\myfunctioncomplexity_p$ are the space and the number of passes required to evaluate \myfunction in the streaming setting, respectively.
\end{theorem}

To prove this theorem, 
it suffices to implement the five problem-independent building blocks of our model-agnostic algorithm (i.e., blocks (1)--(5), cf.~\cref{sec:algorithm:modelagnostic:buildingblocks}) 
with factor $\epsilon^{-1}$ smaller number of passes in the Multi-Pass Streaming model;
we then apply \Cref{thm:metatheorem}, noting that memory can be reused in each iteration.
We will store the $\tildeO(\nnodes)$-sized outputs from Steps~\ref{step:aSSSP}, \ref{step:msf}, \ref{step:rps}, and \ref{step:ffe} explictilty, while handling Steps~\ref{step:ecr} and~\ref{step:cmi} differently to avoid memory cost $\Omega(\nedges)$.

\begin{enumerate}
	\item \emph{Streaming $\alpha$-approximate Set-Source Shortest-Path Forest.}
	\citet{becker21stream} provide a streaming algorithm that, 
	given an undirected graph with non-negative edge weights, 
	computes a $(1 + \epsilon)$-approximate single-source shortest-path tree 
	using $\tildeO(\nnodes)$ space and $\tildeO(\epsilon^{-2})$ passes for any $\epsilon \in (0,1]$. 
	We can convert our \emph{set-source} shortest-path problem into a single-source shortest-path problem by introducing a new node and connecting it to all source nodes by zero-weight edges;
	we simulate this node and its incident edges by holding them in memory, requiring $\BO(n)$ additional words.
	Moreover, we maintain $\forest$ and $\forest'$ in memory, also requiring $\BO(n)$ additional words.
	Note that the lowest node ID (i.e., the component ID) is known for each node as an output of Step~\ref{step:msf} from the preceding iteration (or from initialization). 
	Thus, we can compute (and store) the desired \radius-restricted approximate Set-Source Shortest-Path forest $\forest'$.
	
	\item \emph{Streaming Edge-Cost Reduction.}
	To maintain $\tildeO(\nnodes)$ space, 
	reduced costs of edges are not stored explicitly.
	We exploit that we store, for each node, not only whether it is in the forest, but also its distance to its root.
	Thus, we can infer its reduced cost from the values stored for its endpoints whenever the stream hands us an edge.

	\item \emph{Streaming Candidate-Merge Identification.}
	To maintain $\tildeO(\nnodes)$ space, 
	candidate status of edges is not stored explicitly, either.
	Recall that, in one iteration of the meta-algorithm, an edge $\{\othernode,\node\} \in \edges$ is in \mergeedges if and only if
	\begin{inparaenum}[(i)]
	\item its endpoints are in different trees of the restricted shortest paths forest,
	\item its reduced cost is $0$, and
	\item it has not been deleted from $\edges$.
	\end{inparaenum} 
	The first condition can be checked using the root identifiers (if any) stored by $\othernode$ and $\node$, respectively.
	We already pointed out that reduced costs are stored implicitly, enabling us to check the second condition.
	Regarding the third condition, recall that \Cref{alg:gw-clean} deletes all edges that reached reduced cost $0$  in an iteration except for those in $\forest'\cup \augmentation$.
	By \Cref{lem:gwcleanbasic} (vii), the edges of reduced cost $0$ that are not deleted are exactly those in $\forest\cup \forest'$, 
	i.e., at the beginning of an iteration, $\forest\cup \forest'$ spans exactly the connectivity components of $(\nodes,\forest'\cup \augmentation)$.
	By \Cref{lem:gwcleanbasic} (iii), this is equivalent to spanning $\ballunion$ at the beginning of the iteration (i.e., before the Set-Source Shortest Path forest is computed).
	Therefore, meeting the condition that \othernode and \node lie in different trees of the restricted shortest-path forest implies the third condition, which thus does not need to be checked separately.
	
	\item \emph{Streaming Minimum Spanning Tree.}
	\citet{feigenbaum2005graph} provide a streaming algorithm that, 
	given an undirected graph with non-negative edge weights, 
	computes a Minimum Spanning Tree using $\BO(1)$ passes and $\BO(\nnodes)$ space.
	
	\item \emph{Streaming \stepthree.}
	\citet{becker21stream} provide a streaming algorithm that, 
	given an undirected graph with non-negative edge weights, 
	computes a $(1 + \epsilon)$-approximate Transshipment solution 
	using $\tildeO(\nnodes)$ space and $\tildeO(\epsilon^{-2})$ passes for any $\epsilon \in (0,1]$, 
	i.e., $\tildeO(1)$ passes for $\epsilon= \Theta(1)$. 
\end{enumerate}

\subsection{Applications}

The algorithm presented in the previous section works for \emph{any} proper Constrained Forest Problem in the Multi-Pass Streaming model, 
and hence, 
its running time depends on the complexity of evaluating the forest function \myfunction in the streaming setting. 
Fortunately, evaluating \myfunction in this setting is straightforward in all our applications. 
For the purposes of the following analysis, 
recall that before Forest-Function Evaluation, 
for each node $\node\in\nodes$, 
we already know the identifier of its current component $\component_\node$ from Step~\ref{step:msf} (MST). 

\subsubsection{Steiner Forest}
Like for PRAM, other input representations can be efficiently reduced to SF--IC, with the only non-trivial case being the reduction from SF--CR.
As for PRAM, this is covered by performing an MST computation (followed by simple computations on the memory content), which requires $\BO(1)$ passes and $\BO(\nnodes)$ space~\cite{feigenbaum2005graph}.
Therefore, in the streaming model, we may w.l.o.g.\ assume that the input representation for SF is SF--IC.

We provide an efficient strategy to evaluate the forest function for SF--IC.
For each node $\node\in\nodes$, 
we can learn its label $\componentidentifier(\node)$ in a single pass over the edge stream. 
Storing this information takes $\BO(\nnodes)$ space. 
Without making further passes over the edge stream, 
we can then count  $\nterminals_i = \cardinality{\{\node\in \nodes\mid \componentidentifier(\node) = i\}}$ for each $i\in[\ncomponents]$,  
as well as $\cardinality{\{\node\in\component\mid \componentidentifier(\node) = i\}}$ for each $\component\in\components$ and label $i$ assigned to at least one node in \component. 
As noted in the context of the PRAM model, 
there are at most \nnodes distinct labels, 
and at most \nnodes nonzero $(\component, i, count_i)$ triples, 
such that this information can be stored in $\BO(\nnodes)$ space. 
Using the stored information, 
we can now evaluate, 
for each $\component\in\components$,
whether there exists a label $i\in[\ncomponents]$ such that $\nterminals_i \neq \cardinality{\{\node\in\component\mid \componentidentifier(\node) = i\}} > 0$, 
and record the result as the new activity status of each component, 
which again takes $\BO(\nnodes)$ space. 

\begin{corollary}[Streaming Complexity of Steiner Forest]
	\label{result:sf-mps}
	In the Multi-Pass Streaming model, for any $0<\epsilon\le 1$, Steiner Forest can be $(2+\epsilon)$-approximated in $\tildeO(\epsilon^{-3})$ passes and $\tildeO(\nnodes)$ space.
\end{corollary}

\subsubsection{Point-to-Point Connection}
For each node $\node\in\nodes$, 
set $x_\node = 1$ if \node is a source node, 
$x_\node = -1$ if \node is a target node, 
and $x_\node = 0$ otherwise. 
Then calculate $\sum_{\node\in\component} x_\node$ for each component $\component\in\components$, 
and report all components with a nonzero sum as active components. 
This takes a single pass and $\BO(\nnodes)$ memory words.

\begin{corollary}[Streaming Complexity of Point-to-Point Connection]
	\label{result:ppc-mps}
	In the Multi-Pass Streaming model, for any $0<\epsilon\le 1$, Point-to-Point Connection can be $(2+\epsilon)$-approximated in $\tildeO(\epsilon^{-3})$ passes and $\tildeO(\nnodes)$ space.
\end{corollary}

\subsubsection{\facloclong}

We simulate $s$ and its incident edges using $\BO(n)$ memory,
obtaining the following corollary of the results for SF.

\begin{corollary}[Streaming Complexity of \facloc]
	\label{result:fl-mps}
	In the Multi-Pass Streaming model, for any $0<\epsilon\le 1$, \facloc can be $(2+\epsilon)$-approximated in $\tildeO(\epsilon^{-3})$ passes and $\tildeO(\nnodes)$ space.
\end{corollary}

%% file: text/appendix-lowerbounds.tex
\begin{figure}[t]
	\centering
	\includegraphics[height=7cm]{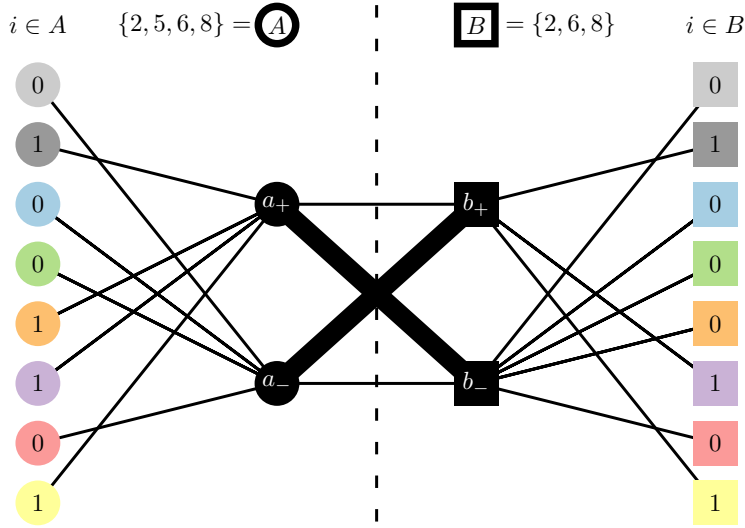}
	\caption{%
		Example of the lower-bound graph used to reduce Equality Testing to Steiner Forest, for an instance with $n = 8$. 
		Node colors indicate input components (SF--CIC) resp. connection requests (SF--SRC), and node shapes indicate which nodes are simulated by which player.}\label{fig:lower-bound-graph}
\end{figure}

\section[Deterministic Lower Bounds for Steiner Forest Input Specifications in Congest]{Deterministic Lower ‌Bounds for Steiner Forest Input Specifications in \Congest}\label{apx:lower}

In this section, we prove that any deterministic algorithm for the Steiner Forest problem with a finite approximation ratio in the \Congest model has time complexity $\tildeOmega(\nterminals)$ (for SF--SCR) and $\tildeOmega(\ncomponents)$ (for SF--CIC), 
leveraging ideas from the work of \citet{lenzen2014steiner}.

\begin{lemma}[Deterministic Lower Bounds for Steiner Forest Input Specifications]\label{lem:sf-lowerbounds}
	In the \Congest model, any deterministic algorithm for the Steiner Forest problem with a finite approximation ratio 
	takes $\tildeOmega(\nterminals)$ rounds if the input is specified via symmetric connection requests (SF--SCR), 
	and $\tildeOmega(\ncomponents)$ rounds if the input is specified as component identifiers supported by component cardinalities (SF--CIC).
\end{lemma}
\begin{proof}
We reduce the equality testing (ET) problem to $\rho$-approximate SF--SCR and $\rho$-approximate SF--CIC as follows. 
Let $A, B \subseteq [n]$ be an instance of ET. 
Alice, who knows $A$, constructs the following (sub)graph, depicted in \Cref{fig:lower-bound-graph}: 
The nodes are the set $\{a_i\}_{i=1}^{n}$, representing each $i \in [n]$, and two additional nodes $a_+$ and $a_-$. 
For each $i \in A$, Alice connects node $a_i$ to $a_+$; for each $i \not \in A$, Alice connects node $a_i$ to $a_-$. 
Formally, the edges are $E_A = \{\{a_+,a_i\} \; | \; i \in A\} \cup \{\{a_-,a_i\} \; | \; i \not \in A\}$. 
Similarly, Bob constructs nodes $\{b_i\}_{i=1}^{n}$, 
representing each $i \in [n]$, 
and two additional nodes $b_+$ and $b_-$. 
For each $i \in B$, Bob connects node $b_i$ to $b_+$; for each $i \not \in B$, Bob connects node $b_i$ to $b_-$. 
Formally, the edges are $E_B = \{\{b_+,b_i\} \; | \; i \in B\} \cup \{\{b_-,b_i\} \; | \; i \not \in B\}$. 
The graph also contains edges $E_{AB} = \{\{a_+,b_+\}, \{a_+,b_-\}, \{a_-,b_+\}, \{a_-,b_-\}\}$. 
All edges have unit cost, except the edges $\{\{a_+,b_-\},\{a_-,b_+\}\}$ which have cost $W \coloneq \rho \cdot (2n + 2) + 1$. 
Finally, we set the connection requests to $\requests_{a_i} = \{b_i\}$ and $\requests_{b_i} = \{a_i\}$ for SF--SCR, 
and the component identifiers to $\componentidentifier(a_i) = \componentidentifier(b_i) = i$ for SF--CIC, 
such that the size of each input component is exactly $2$.

We show that if the algorithm computes a $\rho$-approximation to SF--SCR or SF--CIC on this graph, then we can output the answer ``YES'' to the original ET instance if and only if the algorithm produces an answer that does not include a heavy edge (i.e., $\{a_+,b_-\}$ or $\{a_-,b_+\}$). 
If $A=B$, then all the requests can be satisfied by selecting the edges $E_A \cup E_B \cup \{\{a_+,b_+\},\{a_-,b_-\}\}$. 
If $A \neq B$, then any solution must contain at least one heavy edge (i.e., $\{a_+,b_-\}$ or $\{a_-,b_+\}$), hence it has a weight larger than $\rho \cdot (2n + 2)$.

Now, Alice and Bob can construct the graph based on their input and simulate the algorithm on this graph, i.e., Alice simulates the algorithm on nodes $\{a_i\}_{i=1}^{n} \cup \{a_+,a_-\}$, and Bob simulates the algorithm on nodes $\{b_i\}_{i=1}^{n} \cup \{b_+,b_-\}$. 
The only communication required between Alice and Bob are the messages that cross the edges in $E_{AB}$. 
Solving ET deterministically requires exchanging $\Omega(n)$ bits in the worst case \cite{Kushilevitz_Nisan_1996}. 
In the \Congest model, at most $\BO(\log n)$ bits can cross $E_{AB}$ in each round. 
Thus, the running time of the algorithm must be $\Omega(\nicefrac{n}{\log n})$, 
which is $\tildeOmega(\nterminals)$ for SF--SCR and $\tildeOmega(\ncomponents)$ for SF--CIC.
\end{proof}

%% file: text/appendix-msf.tex
\section[Deterministic Minimum-Spanning-Forest Construction with Partwise Aggregation]{Deterministic Minimum-Spanning-Forest Construction with\newline Partwise Aggregation}
\label{apx:msf}

In this section, we give a deterministic \Congest algorithm to compute the Minimum Spanning Forest (and, consequently, also a Minimum Spanning Tree) in $\tildeO(\pacomplexity)$ time, 
derandomizing the algorithm by \citet{ghaffari2016algorithms}. 
Our algorithm is shown as \cref{alg:detmsf}.

\begin{lemma}[Deterministic MSF Computation]
	In the \Congest model, a Minimum Spanning Forest can be computed deterministically in $\tildeO(\pacomplexity)$ time.
	This holds also in the presence of $\BO(1)$ virtual nodes.
\end{lemma}

\begin{algorithm2e}[t]
	\DontPrintSemicolon
	\caption{Deterministic Minimum Spanning Forest computation in $\tildeO(\pacomplexity)$ time.}\label{alg:detmsf}
	\KwIn{Graph $\graph = (\nodes,\edges)$, edge costs $\cost\colon\edges\rightarrow\naturals$}
	\KwOut{Minimum Spanning Forest \forest}
	$\forest \gets \emptyset$ \;
	$\components \gets \{\{\node\}\mid\node\in\nodes\}$\;
	\ForEach{$i\in[\lceil\log\nnodes\rceil]$}{
		$\forest'\gets \emptyset$\;
		\ForEach{$\component \in \components$}{
			Determine the lightest edge leaving $\component$ and add it to $\forest'$\;
		}
		Add a maximal matching $\forest_M \subseteq \forest'$ in the graph $(\components,\forest')$ to \forest\;
		If $\component \in \components$ has no incident edge in $\forest_M$, it adds the edge it selected into $\forest'$ to \forest\;
		Merge components connected in \forest to update \components\;
	}
	\Return \forest
\end{algorithm2e}

\begin{proof}
	Consider first the special case without virtual nodes.
	Similar to \citet{ghaffari2016algorithms}, 
	we have $\BO(\log \nnodes)$ phases in total. 
	Denote by $\components$ the set of connectivity components so far.
	Nodes locally store the smallest identifier of a node in their component as the component ID.
	Starting with each node in its own component, in each phase, each component $\component \in \components$ computes its minimum-weight outgoing edge (breaking ties by identifiers). 
	This can be done in $\BO(\pacomplexity)$ time by aggregating the minimum over the weights of outgoing edges witnessed by nodes in \component.
	As in the algorithm, denote the set of these edges by $\forest'$.
	Note that $\forest'$ is part of the MSF, since a lightest outoing edge can replace a heavier outgoing edge in a spanning tree to obtain a spanning tree of smaller weight.
	
	We now simulate the Cole-Vishkin algorithm \cite{cole1986coloring} to compute a $3$-coloring of $(\components,\forest')$ in $\mathcal{O}(\pacomplexity \cdot \log^*\nnodes)$ time.
	It suffices to send the color of the component's parent to the nodes contained in the component, 
	which is achieved in $\BO(\pacomplexity)$ rounds using Partwise Aggregation. 
	Next, we select a maximal matching and add it to the MSF in $\BO(\pacomplexity)$ time by sequentially going over the color classes. 
	When coming to color $c$, all unmatched components with color $c$ will first check for an unmatched incoming merge, and select one if there are any (breaking ties by identifiers), requiring $\BO(\pacomplexity)$ rounds.
	For each unmatched component, we add the edge it previously selected into $\forest'$ to the MSF (also in $\BO(\pacomplexity)$ rounds).
	 
	Finally, it remains to determine and distribute the new component identifier within each newly formed connected component.
	To this end, observe that contracting the old components results in a graph in which each connected component has diameter at most $3$: 
	It contains at most one matching edge, and any selected non-matching edge has exactly one matched endpoint.
	Thus, we can determine and make known the lowest node identifier within each new connected component by $\BO(1)$ Partwise Aggregations and additional communication rounds over the newly selected edges. 

	Overall, each iteration takes $\BO(\pacomplexity\log^* n)$ rounds, and iterating $\lceil \log\nnodes \rceil$ times guarantees that all connected components are spanned: 
	Each component that is not identical to a connected component of the input graph gets merged with at least one other component in each iteration, at least doubling the minimum size of a non-spanning component in each step.
	As all edges we add belong to the MSF,
	the output forest is indeed the MSF of the graph.
	
	To complete the proof, we need to show that the same complexity can be achieved if some of the nodes are virtual.
	The states of the virtual nodes are maintained by all nodes.
	When determining lightest outgoing edges, i.e., $\forest'$, we perform a (global) aggregation for each virtual node, adding $\BO(\hopdiameter)$ rounds, followed by a broadcast to update its state at all nodes.
	The same applies to all other steps of the algorithm.
	Thus, essentially the same procedure can be used to implement each of the steps of an iteration with virtual nodes, adding in total $\tildeO(\hopdiameter)\subseteq \tildeO(\pacomplexity)$ rounds.
\end{proof}

%% file: text/appendix-related.tex
\section{Further Related Work}\label{apx:related}

Beyond the related work discussed in the main text, 
here, we briefly summarize 
relevant research related to the \emph{methodological foundations} of our framework, 
i.e., primal-dual approaches, 
as well as relevant literature related to our \emph{class of problems}, 
i.e., survivable network design problems. 
Given a graph $\graph = (\nodes,\edges)$ with edge costs $\cost(\edge)$, 
\emph{survivable network design problems} (SNDPs) ask us to select a minimum-cost edge subset satisfying any given node- or edge-connectivity requirements. 
Since SNDPs are both NP-hard in general and relevant in practice, 
their approximability has been studied extensively for specific problems,  
graph classes, and computational models. 

\emph{Primal-Dual Approaches.}\quad
In the classic \emph{sequential setting}, 
SNDPs saw one of the earliest applications of the primal-dual method for approximation algorithms, 
yielding $(2-\nicefrac{2}{\nterminals})$-approximations for Steiner Forest and other Constrained Forest Problems with proper forest functions  \cite{agrawal1995trees,goemans1995approximation,goemans1996primal}, 
as well as a $3$-approximation for $2$-ECSS \cite{klein1993cycles} and an $\BO(\log k)$-approximation for $k$-ECSS \cite{williamson1995primal,goemans1994improved}. 
In the \emph{distributed setting}, 
primal-dual approaches have been leveraged to tackle problems like metric facility location \cite{pandit2009return}, 
minimum dominating set \cite{pandit2010rapid},
load balancing \cite{ahmadian2021distributed}, 
approximate weighted matching \cite{panconesi2010fast},  
and approximate $k$-core decomposition \cite{chan2021distributed}. 
Notably, by porting the primal-dual algorithm of \citet{agrawal1995trees} to the distributed setting, 
the Steiner Forest algorithm presented by \citet{lenzen2014steiner} also (implicitly) follows a primal-dual approach. 
Beyond the models considered in our work, 
in the \emph{online setting}, 
primal-dual methods have been combined with ideas from online set cover \cite{alon2003online,gupta2012online} 
to design algorithms for fractional network design problems \cite{alon2006general,buchbinder2006improved}.

Most primal-dual results on network design problems have been limited to functions that are \emph{proper} (in particular, satisfying \emph{disjointness}) 
or at least \emph{uncrossable} 
(i.e., $\myfunction(A) = \myfunction(B) = 1$ implies that either $\myfunction(A\cap B) = \myfunction(A\cup B) = 1$ 
or $ \myfunction(A\setminus B) = \myfunction(B\setminus A) = 1$, which strictly generalizes disjointness). 
Optimizing these functions comes with the convenience that there exists an optimal dual solution with laminar support i.e., the sets associated with positive dual-variable values are either nested or disjoint. 
This might allow us to extend our framework to uncrossable functions, 
as posed as an open question in \cref{sec:conclusion}. 
Moreover, recent work obtains constant-factor approximations 
even for network design problems specified by functions that are \emph{not} uncrossable and come without laminar optimal-dual support \cite{bansal2023improved}. 
Hence, further investigating which classes of functions allow for good model-agnostic primal-dual approximations appears as a promising avenue for future~work. 

\emph{Survivable Network Design Problems.}\quad
Since the literature on SNDPs is incredibly vast, 
we discuss only the work immediately addressing our example problems.

For \emph{Steiner Forest} (SF), in the \emph{sequential setting}, 
\citet{garg2002fast} develop a fully polynomial-time approximation scheme for \emph{fractional} SF,
\citet{borradaile2015ptas} design a polynomial-time approximation scheme (PTAS) for SF in Euclidean metrics, 
and \citet{chan2018ptas} give a PTAS for SF in doubling metrics (which include Euclidean metrics as a special case).
Furthermore, 
\citet{bateni2011approximation} design a  PTAS for SF on planar graphs and graphs of bounded treewidth. 
In the \emph{parallel setting}, \citet{suzuki1990parallel} study SF in \emph{unweighted planar} graphs,  
and in the \emph{streaming setting}, \citet{czumaj2022streaming} design approximation algorithms for \emph{geometric} SF---%
i.e., both works consider settings that are more restricted than ours. 
In the \emph{distributed setting}, 
our work significantly improves over the state of the art established by \citet{lenzen2014steiner}.

While \emph{Point-to-Point Connection}, 
studied in the \emph{sequential setting} already by \citet{li1992delivery},
is one of the original Constrained Forest Problems considered by \citet{goemans1995approximation},
to the best of our knowledge, no prior work has studied Point-to-Point Connection in any of the models we consider. 

The \emph{\facloclong} (\facloc) problem is different from the classic Uncapacitated Facility Location (UFL) problem   
in that it minimizes the sum of facility-opening and \emph{network-construction costs}, rather than \emph{connection costs} (i.e., client-to-facility distances).
Our specification is motivated by the need, identified in operations research, 
to capture route-sharing scenarios, 
in which the costs of clients cannot be considered independently \cite{klose2005facility},
and it can be naturally extended to include the \emph{connection costs} from each client to its closest selected facility using the selected forest  \cite{melkote2001integrated}, 
as posed as an open question in \cref{sec:conclusion}. 
For \emph{non-metric UFL}, 
the $\BO(\log c)$-approximation (where $c$ is the number of clients) obtained by \citet{hochbaum1982heuristics} in the \emph{sequential setting} is asymptotically optimal, 
as approximating set cover (a special case of non-metric UFL) to within factor $(1 - \epsilon)\ln \nnodes$ is NP-hard for every $\epsilon > 0$ \cite{dinur2014analytical}. 
In contrast, for \emph{metric UFL}, 
\citet{li2013ufl} obtain a 1.488-approximation by merging techniques from \citet{byrka2010optimal} and \citet{mahdian2006approximation}. 
Recalling that the cost of facility-location problems can be split into node-related costs $o$ (opening facilities) and edge-related costs $d$ (connecting facilities), 
an algorithm yields a \emph{Lagrangian-multiplier-preserving} $\alpha$-approximation if its solution $S$ is guaranteed to satisfy, for any feasible solution $S^*$, $o(S) + d(S) \leq o(S^*) + \alpha d(S^*)$. 
Here, in recent work, 
\citet{cohen2023breaching} give the first Lagrangian-multiplier-preserving $\alpha$-approximation for $\alpha < 2$ by combining dual fitting with local search.

Both non-metric and metric UFL have also been studied in the \emph{distributed setting}. 
Specifically in the \Congest model, 
\citet{moscibroda2005facility} provide an algorithm that, 
for every positive integer $k$, 
yields a $\BO(k(fc)^{1/k}\log (f + c))$-approximation for \emph{non-metric UFL} in $\BO(k)$ rounds (where $f$ is the number of facilities), 
and for \emph{metric UFL}, 
\citet{pandit2009return} give a $7$-approximation in $\BO(\log f + \log c)$ rounds using a primal-dual approach based on the work of  \citet{jain2001approximationb}.
In the \emph{parallel setting}, 
\citet{blelloch2010parallel} give a $(3 + \epsilon)$-approximation for \emph{metric UFL} in $\BO(\nedges \log_{1+\epsilon} \nedges)$ work and $\tildeO(1)$ depth.
Finally, in the \emph{streaming setting}, 
\citet{czumaj2022streaming,czumaj2013approximation} study \emph{Euclidean uniform UFL} (i.e., UFL on the grid $\{1,\dots,\Delta\}^d$ with uniform facility-opening costs), 
providing an $(1 + \epsilon)$-approximation for $d = 2$ using one pass and $(\frac{\log\Delta}{\epsilon})^{\BO(1)}$ space \cite{czumaj2013approximation} and, 
more recently, a $d^{3/2}$-approximation using one pass and $\tildeO(d\log \Delta)$ space \cite{czumaj2022streamingb}.